\documentclass[conference,compsoc]{IEEEtran}

\ifCLASSOPTIONcompsoc
  \usepackage[nocompress]{cite}
\else
  \usepackage{cite}
\fi

\ifCLASSINFOpdf
\else
\fi

\usepackage{caption}

\usepackage[linesnumbered,ruled]{algorithm2e}
\usepackage{url}
\usepackage{balance}
\usepackage{enumitem}
\usepackage{tcolorbox}
\usepackage{hyperref}
\usepackage{diagbox}
\usepackage{xspace}
\usepackage{amsmath,amssymb,amsthm}
\usepackage{booktabs}
\usepackage{multirow}
\usepackage{makecell}
\usepackage{float}
\usepackage{authblk}
\usepackage{cleveref}
\usepackage{graphicx}
\usepackage{subcaption}
\usepackage{caption}
\usepackage{subcaption}
\usepackage{colortbl}  
\usepackage{xcolor, eucal}
\usepackage{array}  
\definecolor{gray0}{gray}{0.9}

\newtheorem{theorem}{Theorem}[section]
\newtheorem{definition}{Definition}[section]

\begin{document}
\pagestyle{plain}

\def \toolname{DP-GCL\xspace}

\title{\Large \bf Differentially Private Contrastive Learning via Bounding Group-level Contribution}

\author{Kecen Li\textsuperscript{1}, Chen Gong\textsuperscript{2,*}, Zinan Lin\textsuperscript{3}, Tianhao Wang\textsuperscript{2}, and Xiaokui Xiao\textsuperscript{1}\\
\textsuperscript{1}National University of Singapore\quad \textsuperscript{2}University of Virginia\quad \textsuperscript{3}Microsoft Research}

\maketitle
{\makeatletter\def\@makefntext#1{\noindent#1}\makeatother\footnotetext{\textsuperscript{*}Corresponding author.}}

\begin{abstract}
Differentially private (DP) contrastive learning aims to learn general-purpose representations from sensitive data, alleviating privacy leakage concerns associated with embedding models trained on private user content. However, standard contrastive objectives compare each sample with all other samples in the batch, so a single record affects the gradient of every sample in the batch. These inter-sample dependencies make it difficult to bound each record's influence while retaining useful training signals under DP.

We propose DP-GCL, a DP contrastive learning framework that jointly localizes contrastive dependencies and bounds the resulting gradients. Its core module, Bounding Group-level Contribution (BGC), partitions each batch into disjoint groups, restricts comparisons to within-group samples, and clips each group's gradient separately before aggregation. This design confines each record's influence to one group while allowing signals from multiple groups to accumulate under a sensitivity bound independent of the number of groups. Because grouping reduces the available negatives, we further introduce intra-group sample augmentation to enrich local negative views without additional privacy cost.

Across eight datasets, DP-GCL improves image classification and image-text retrieval accuracy over the strongest baseline by an average of 5.6 and 20.0 percentage points, respectively. Further ablation studies indicate that these improvements arise primarily from the joint design of group-local contrastive objectives and group-level gradient clipping. Moreover, DP-GCL exhibits favorable batch-size scaling and remains compatible with diverse network architectures and parameter-efficient fine-tuning strategies such as LoRA.
\end{abstract}

\section{Introduction}
\label{sec:introduction}

Contrastive learning~\cite{ContrastiveLearning} learns transferable representations by bringing positive pairs closer in embedding space and separating negative pairs. Positive pairs can be different augmentations of the same image or an image and its text description. Methods such as SimCLR~\cite{simclr}, MoCo~\cite{moco}, and CLIP~\cite{CLIP} support applications ranging from content recommendation~\cite{cl_recom1} to medical image analysis~\cite{cl_medi1}. However, these models can leak information about their training data~\cite{attackcl1,attackcl2,pan2020privacy,liu2022stolenencoder,attackcl4,hu2025vlmmia}. For example, membership inference attacks exploit statistical differences in embedding similarities to infer whether a sample was used during training~\cite{attackcl2,attackcl3}. The goal is therefore to learn transferable representations while protecting the training records.

Differential privacy (DP)~\cite{dp} offers a principled way to limit privacy leakage by ensuring that changing one record does not significantly change the distribution of the released model. In supervised learning where losses decompose over records (e.g., classification), DP-SGD~\cite{dpsgd} typically achieves this by clipping each sample's gradient and adding calibrated noise~\cite{dp_classification1,dp_classification2}. Applying this approach to contrastive learning is more challenging. 
Standard batch-based contrastive objectives (detailed in~\Cref{subsec:dpcl}) violate this decomposition: a record participates in other records’ comparisons, so changing it can alter every per-sample gradient in the batch. Clipping each gradient separately therefore does not bound the record’s total influence by the same clipping  threshold.

Existing methods adapt gradient clipping to accommodate inter-sample dependencies while retaining contrastive comparisons across the entire batch. DP-CLIP~\cite{batch-level} clips the entire batch gradient as a single contribution, bounding each record's influence but also limiting how much signal can accumulate as more samples are combined. Such accumulation is important for realizing the utility benefits of large-batch DP training~\cite{imagenet_dp,de2022unlocking}. Logit-DP~\cite{pair-level} instead clips pairwise similarity gradients before aggregation, which lets multiple contributions accumulate. However, each record still takes part in comparisons across the entire batch, so the noise must be calibrated to a larger sensitivity bound than under batch-level clipping. 
These two designs show that the difficulty goes beyond choosing what to clip. Signal accumulates only when a batch contributes many separately clipped gradients, yet with batch-wide comparisons a record affects all of them, which inflates the bound on its influence. Shrinking this bound requires removing some of a record's comparisons, although these comparisons are what contrastive learning learns from. The comparisons to keep and the gradients to clip therefore have to be designed together.

To this end, we propose \toolname{}, a DP contrastive learning framework with two modules. \textit{Bounding Group-level Contribution (BGC)} partitions each batch into disjoint groups, computes a contrastive objective using only the samples within each group, and separately clips the group gradients before aggregating them with DP noise. Because each record belongs to exactly one group, it affects only one clipped contribution. The sensitivity bound is therefore independent of the number of groups, while contributions from multiple groups can still accumulate.  Grouping also reduces the number of available negatives, which can weaken contrastive learning~\cite{yuan2022provable,DBLP:conf/nips/ChenZXCD0TZC22}. To mitigate this cost, \textit{Intra-Group Sample Augmentation (ISA)} generates additional negative views within each group, enriching comparisons without increasing the privacy cost.

We evaluate \toolname{} on uni-modal and multi-modal contrastive learning tasks. Across eight datasets and privacy budgets $\epsilon\in\{1,10\}$, the mean absolute improvement over DP-CLIP~\cite{batch-level} is 5.6 percentage points for image classification and 20.0 percentage points for image-text retrieval. Our experiments also show favorable batch-size scaling and more stable training trajectories. Ablations support the joint design of group-local objectives and group-level clipping as the main source of the observed gains, with ISA providing additional improvements. We further evaluate \toolname{} with LoRA, several ViT backbones, and different grouping and augmentation settings. In the group-size experiments, moderate expected group sizes (approximately 4--16) provide the best observed trade-offs.

\vspace{1mm}
\noindent \textbf{Contributions.}
We summarize our main contributions as follows:
\begin{itemize}[leftmargin=*]
    \item We propose \toolname{}, a DP contrastive learning framework that confines each record's influence on the gradient update to one group by jointly designing group-local contrastive objectives and group-level gradient clipping. ISA further enriches within-group negative views without increasing the privacy cost.

    \item We establish record-level privacy guarantees of \toolname{} with a sensitivity bound independent of the number of groups, and derive a convergence bound that shows how grouping trades off negative-sample availability against clipped gradient accumulation.

    \item We conduct comprehensive evaluations across eight datasets in both uni-modal and multi-modal settings, demonstrating substantial improvements over existing methods under varying privacy budgets.

    \item We investigate the scalability and flexibility of \toolname{} across different batch sizes, model architectures, and parameter-efficient fine-tuning configurations, including LoRA. Our results demonstrate favorable batch-size scaling and broad architectural compatibility. We release our replication package at an anonymous link.\footnote{\url{https://anonymous.4open.science/r/DP-GCL-A952/}}
\end{itemize}

\begin{figure}[t]
    \centering
    \includegraphics[width=0.95\linewidth]
    {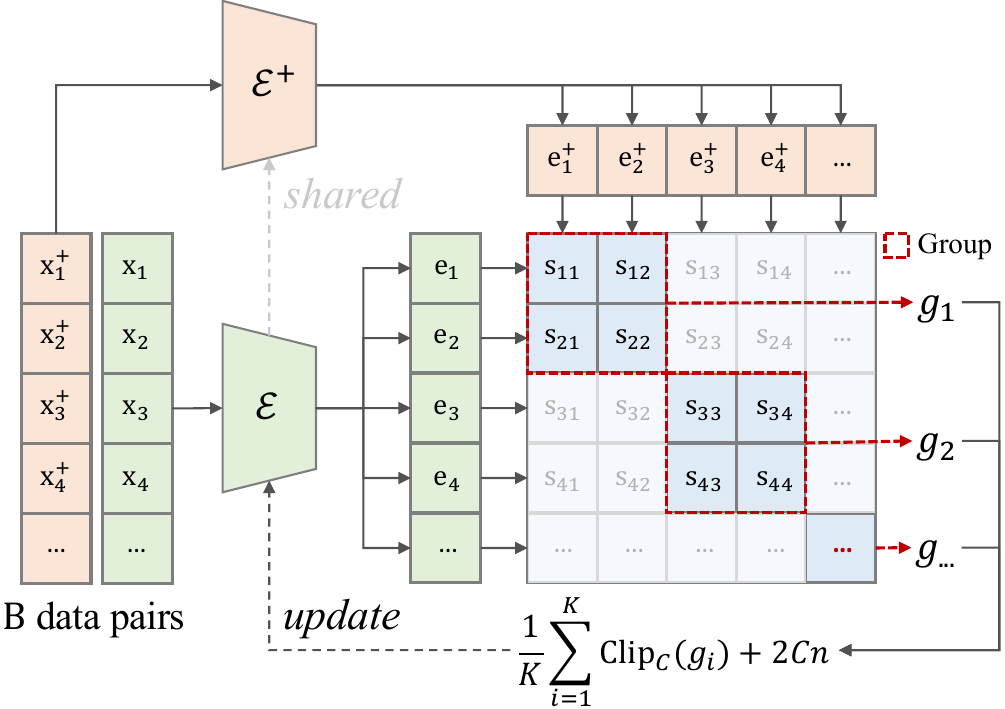}
    \caption{Overview of \toolname{}. Given a
    Poisson-sampled batch of positive pairs, we
    randomly partition the pairs into $K$ disjoint
    groups. Within each group, contrastive
    interactions are restricted to local samples,
    eliminating cross-group dependencies. ISA
    further generates augmented negative views
    within each group. The group-local contrastive
    gradients are independently clipped and
    aggregated with Gaussian noise to achieve DP.
    For uni-modal contrastive learning, the
    parameters of $\mathcal{E}$ and $\mathcal{E}^+$
    are shared.}
    \label{fig:framework2}
\end{figure}

\section{Background}
\label{sec:background}

\subsection{Differential Privacy}
\label{subsec:dp}

\begin{definition}[($\epsilon, \delta$)-Differential Privacy~\cite{dp}]
\label{def:dp}
A randomized algorithm $\mathcal{M}$ satisfies $(\epsilon, \delta)$-differential privacy for parameters $\epsilon > 0$ and $\delta \geq 0$ if, for every pair of adjacent datasets $\mathcal{D}$ and $\mathcal{D}'$, and for all measurable subsets $O$ of the output space of $\mathcal{M}$, the following inequality holds:
\begin{equation}
    \Pr[\mathcal{M}(\mathcal{D}) \in O] \leq e^\epsilon \cdot \Pr[\mathcal{M}(\mathcal{D}') \in O] + \delta.
\end{equation}
\end{definition}

\noindent We use add-or-remove adjacency: $\mathcal{D} \simeq \mathcal{D}'$ if the datasets differ by one record. Smaller $\epsilon$ provides stronger privacy, while $\delta$ controls the additive slack in the guarantee.

\noindent \textbf{Gaussian Mechanism.}
For a function $f$ with $\ell_2$ sensitivity $\Delta_f = \max_{\mathcal{D} \simeq \mathcal{D}'} \|f(\mathcal{D})-f(\mathcal{D}')\|_2$, the Gaussian mechanism adds noise calibrated to $\Delta_f$:
\begin{equation}\label{eq:gm}
    \mathcal{M}(\mathcal{D}) = f(\mathcal{D}) + \mathcal{N}\left(0, \Delta_f^2 \sigma^2 \mathbb{I}\right),
\end{equation}
where $\sigma>0$ is the noise scale governing the privacy guarantee~\cite{dp}.

\noindent \textbf{DP-SGD.}
DP-SGD~\cite{dpsgd} clips per-sample gradients, aggregates them, and adds calibrated Gaussian noise before updating the model parameters $\theta$:
\begin{align*}
\theta \gets \theta - \lambda \left( \sum_{x_i \in b} \operatorname{Clip}_C\big(\nabla_{\theta} \mathcal{L}(x_i)\big) + \mathcal{N}\big(0,\, C^2 \sigma^2 \mathbb{I}\big) \right),
\end{align*}
where $\lambda$ is the learning rate, $b$ is a sampled mini-batch, and $\operatorname{Clip}_C(\cdot)$ bounds each gradient's $\ell_2$ norm by $C$. For decomposable losses, each record affects only its own gradient, so the clipped gradient sum has sensitivity at most $C$ under add-or-remove adjacency.

\subsection{DP Contrastive Learning}
\label{subsec:dpcl}

Contrastive learning~\cite{ContrastiveLearning} learns representations by pulling positive pairs together and pushing negative pairs apart. Let $\mathcal{B}=\{(x_i,x_i^+)\}_{i=1}^B$ be a batch of positive pairs, such as two views of an image or an image--text pair. We measure similarity using cosine similarity:
\begin{equation}
\label{eq:sim}
    s_{i,j} = \frac{\mathcal{E}(x_i)^\top\mathcal{E}(x_j^+)}{\|\mathcal{E}(x_i)\|_2 \|\mathcal{E}(x_j^+)\|_2},
\end{equation}
where $\mathcal{E}$ is an embedding model parameterized by $\theta$. Following prior DP contrastive learning work~\cite{pair-level,batch-level}, we focus on the widely adopted InfoNCE objective~\cite{simclr,CLIP}:
\begin{equation}
\label{eq:simclr}
    \mathcal{L}_\text{NCE} = \sum_{i=1}^B-\log \frac{e^{{s_{i,i}}/{\tau}}}{\sum_{j=1}^B e^{{s_{i,j}}/{\tau}}},
\end{equation}
where $\tau$ is the temperature parameter. The $i$-th term depends on all records in the batch through its denominator, introducing global inter-sample dependencies into gradient computation.

We next examine how sample-level and batch-level clipping handle these dependencies, highlighting their trade-offs between sensitivity and gradient signal accumulation.

\subsubsection{Sample-level Bounding}  
A straightforward approach to designing a DP contrastive learning algorithm is to directly apply standard DP-SGD~\cite{dpsgd} to the InfoNCE loss. In this formulation, gradients are computed per sample and then clipped to bound contributions,
\begin{equation}
\label{eq:f_sample}
    g_\text{sample} = \sum_{i=1}^B \text{Clip}_C\Biggl(\nabla_\theta -\log \frac{e^{{s_{i,i}}/{\tau}}}{\sum_{j=1}^B e^{{s_{i,j}}/{\tau}}}\Biggr).
\end{equation}
Sample-level bounding lets multiple clipped contributions accumulate. For global InfoNCE, adding a record can change all $B$ existing contributions and introduce one new contribution. The triangle inequality therefore gives a bound of $(2B+1)C$ on the difference between the sums for batches of sizes $B$ and $B+1$~\cite{pair-level}. Unlike decomposable losses, this bound grows with batch size, and calibrating noise to it can limit utility.

\subsubsection{Batch-level Bounding}  
To mitigate this high sensitivity, Huang et al.~\cite{batch-level} proposed DP-CLIP, which instead applies clipping after aggregating gradients across the batch,
\begin{equation}
\label{eq:f_batch}
    g_\text{batch} = \text{Clip}_C\Biggl(\sum_{i=1}^B \nabla_\theta -\log \frac{e^{{s_{i,i}}/{\tau}}}{\sum_{j=1}^B e^{{s_{i,j}}/{\tau}}}\Biggr).
\end{equation}
Batch-level bounding gives a sensitivity bound of $2C$, independent of batch size, but also bounds the entire signal norm by $C$. At a fixed noise scale, a larger batch therefore cannot raise this ceiling relative to the noise, whereas separately clipped contributions can add up to more than $C$.

These limitations motivate jointly localizing contrastive dependencies and aligning them with fine-grained clipping units to achieve low sensitivity while preserving gradient signal accumulation. We develop this design in~\Cref{sec:method}.

\subsection{Threat Model} \label{subsec:threat_model}

We consider a trusted curator that trains contrastive encoders on sensitive data and releases them or provides an embedding service. The threat surface comprises the released parameters and outputs, while the private data and intermediate states remain inaccessible. An external adversary has white-box access, knows the training algorithm, and may use arbitrary auxiliary information. This models practical encoder-sharing and embedding-service deployments and covers weaker output-only access. The adversary aims to infer whether a target record participated in training or recover record-specific information, consistent with existing attacks~\cite{attackcl1,attackcl2,attackcl3,carlini2022lira,balle2022reconstructing}. \toolname{} provides attack-agnostic, record-level \((\epsilon,\delta)\)-DP under add-or-remove adjacency, where one record is exactly one positive pair, such as two views of an image or an image–text pair. The guarantee is preserved under post-processing. It does not provide the same guarantee at the user or entity level when one individual contributes multiple records, or protect against a compromised curator or exposed training states.

\section{Methodology}
\label{sec:method}

This section presents \toolname{}, which controls each record's influence while enabling gradient signal accumulation by jointly localizing contrastive comparisons and clipping the resulting group gradients. Bounding Group-level Contribution (BGC) implements this joint design, and Intra-Group Sample Augmentation (ISA) enriches the negative views within each group without additional privacy cost.

Following previous work~\cite{pair-level},
we adopt the add-or-remove notion of neighboring
datasets in DP~\cite{dp}. Specifically, given a
dataset $\mathcal{D}=\{(x_i,x_i^+)\}_{i=1}^{N}$,
two datasets are neighboring if they differ by
the addition or removal of one positive pair.
For uni-modal learning, $x_i$ denotes an image and
$x_i^+$ is its augmented view, while for multi-modal
learning, $x_i$ denotes an image and $x_i^+$ is its
corresponding text description.

\subsection{Motivation}
\label{subsec:motivation}

Following the trade-off discussed in~\Cref{sec:introduction}, we seek to retain useful contrastive comparisons while allowing multiple clipped gradients to accumulate under a bounded sensitivity. Realizing this idea requires addressing two related design challenges.

\vspace{1mm}
\noindent \textbf{C.1: Localizing inter-sample dependencies.}
Global InfoNCE allows a record to influence losses throughout the batch. As discussed in~\Cref{subsec:dpcl}, a record appears in the denominator of every term in~\Cref{eq:simclr}, so adding it changes all $B$ per-sample gradients, and sample-level bounding calibrates noise to the bound of $(2B+1)C$. We need to confine this influence to a smaller set of records while retaining useful positive and negative comparisons. This motivates restricting contrastive interactions to disjoint groups, so that a record appears only in the terms of its own group. The cost is that each record is compared with fewer negatives, which we analyze in~\Cref{subsubsec:group_clipping}.

\vspace{1mm}
\noindent \textbf{C.2: Preserving gradient accumulation under bounded sensitivity.}
Localizing comparisons alone does not ensure that signals can accumulate after clipping, which is important for large-scale DP training~\cite{imagenet_dp}. Summing all group gradients before clipping still leaves only one bounded contribution, as in batch-level bounding, whose signal norm is capped at $C$ regardless of the batch size. Conversely, clipping groups separately while retaining global negatives creates a dependency--clipping mismatch: one record can still affect multiple clipped contributions, and the sensitivity bound then grows with the number of groups. The challenge is therefore to accumulate multiple bounded contributions while limiting how many of them each record can affect.

BGC addresses these design challenges by using the same partition for contrastive comparisons and gradient clipping. Each record then influences exactly one clipped contribution, which keeps the sensitivity at $2C$, while the separately clipped contributions of all groups can still accumulate. \Cref{fig:framework2} shows the resulting workflow: Poisson sampling, independent assignment to groups, group-local comparisons enriched by ISA, and aggregation of separately clipped gradients with Gaussian noise.

\subsection{Bounding Group-level Contribution (BGC)}
\label{subsec:group_bounding}

BGC builds on microbatch gradient clipping~\cite{mcmahan2018general,ponomareva2023dpfy}, which clips contributions from subsets of a batch before aggregation. We apply this principle to contrastive learning by jointly choosing the group-local objectives and their clipping units, as detailed below.

\subsubsection{Group-local Contrastive Objective}

To address C.1, we restrict contrastive interactions to disjoint groups, confining each record's influence to the objective of its own group. Specifically, given a sampled batch
$\mathcal{B}=\{(x_i,x_i^+)\}_{i=1}^{B}$,
we independently assign each positive pair
to one of $K$ groups using random marks
$r_i\sim\operatorname{Unif}\{1,\ldots,K\}$.
The resulting groups are defined as
\[
\mathcal{G}_k=
\{i\in\{1,\ldots,B\}:r_i=k\}.
\]
Conditioned on the sampled batch size $B$, the expected group size is $B/K$, and the realized group sizes may differ.

Based on this partition, we modify the
standard InfoNCE objective by restricting
the available negative samples to those
within the same group:
\begin{equation}
\label{eq:simclr_group}
\mathcal{L}_{\mathrm{NCE}}^{\mathrm{group}}
=
\sum_{k=1}^{K}
\sum_{i\in\mathcal{G}_k}
-\log
\frac{e^{s_{i,i}/\tau}}
{\sum_{j\in\mathcal{G}_k}e^{s_{i,j}/\tau}}.
\end{equation}

Unlike the original InfoNCE objective,
where every sample may interact with the
entire batch, the group-local objective
only depends on samples assigned to the
same group. Consequently, adding or
removing one record changes the contrastive
objective of at most one group, while
the remaining group objectives are unaffected.

This localization removes cross-group dependencies while retaining the positive and negative comparisons available within each group. However, localizing the objective alone
does not fully resolve the limitations
of DP gradient computation. If the gradients
from all groups are aggregated and then
clipped as a single batch-level contribution,
the resulting update is still restricted
by one clipping threshold.
Therefore, we further introduce group-level
gradient clipping to preserve fine-grained
gradient signal accumulation.

\subsubsection{Group-level Gradient Clipping}
\label{subsubsec:group_clipping}

To enable gradient accumulation across
multiple clipping units, we consider
group-level gradient clipping, where
gradients from different groups are
independently clipped before aggregation.

However, applying group-level clipping
directly to the original global InfoNCE
objective does not eliminate cross-group
dependencies. Specifically, the resulting
gradient estimator is
\begin{equation}
\label{eq:f_group_clipping}
g_{\mathrm{g\_clip}}
=
\sum_{k=1}^{K}
\operatorname{Clip}_C
\Biggl(
\nabla_\theta
\sum_{i\in\mathcal{G}_k}
-\log
\frac{e^{s_{i,i}/\tau}}
{\sum_{j=1}^{B}e^{s_{i,j}/\tau}}
\Biggr),
\end{equation}
where $\operatorname{Clip}_C(\cdot)$ bounds
the $\ell_2$ norm of each group contribution
by $C$.

Although gradients are clipped separately
for each group, the denominator of the
original InfoNCE loss still contains
samples from the entire batch.
Consequently, changing a single record
may affect the clipped gradients of
all $K$ groups, leading to the following
sensitivity bound.

\begin{theorem}
\label{the:group_clipping}
Under add-or-remove adjacency, the
unnormalized gradient estimator
$g_{\mathrm{g\_clip}}$ in
\Cref{eq:f_group_clipping}
has $\ell_2$ sensitivity bounded by
\(\Delta_{g_{\mathrm{g\_clip}}}
\leq 2KC.\)
\end{theorem}

\noindent
A record can affect all $K$ clipped contributions. Each contribution changes by at most $2C$, so the triangle inequality gives the bound of $2KC$. An explicit proof is provided in Appendix~\ref{apsec:MissingProofs}.
With global negatives, a record's dependencies extend beyond its clipping unit, illustrating the mismatch in C.2. The bound in~\Cref{the:group_clipping} scales with $K$, so calibrating noise to this bound can offset the benefit of accumulating more group contributions.

To overcome this limitation, we jointly
employ the group-local contrastive objective
in~\Cref{eq:simclr_group} and group-level
gradient clipping. The resulting
gradient estimator is
\begin{equation}
\label{eq:f_group}
g_{\mathrm{group}}
=
\sum_{k=1}^{K}
\operatorname{Clip}_C
\Biggl(
\nabla_\theta
\sum_{i\in\mathcal{G}_k}
-\log
\frac{e^{s_{i,i}/\tau}}
{\sum_{j\in\mathcal{G}_k}e^{s_{i,j}/\tau}}
\Biggr).
\end{equation}

For this group-local InfoNCE objective,
empty and singleton groups contribute
zero gradients.

The contrastive dependency
structure is now aligned with the clipping
units. Each record influences only the
gradient of its assigned group, while
the clipped contributions from other
groups remain unchanged. Therefore,
the sensitivity no longer grows with $K$.

\begin{theorem}
\label{the:group_grad}
Under add-or-remove adjacency, the
unnormalized group-local gradient estimator
$g_{\mathrm{group}}$ in
\Cref{eq:f_group}
has $\ell_2$ sensitivity bounded by
\(\Delta_{g_{\mathrm{group}}}
\leq 2C,\)
independent of the group count $K$.
\end{theorem}

\noindent
Each record changes at most one clipped contribution, whose difference between neighboring executions has norm at most $2C$. All other contributions cancel. The explicit proof is given in Appendix~\ref{apsec:MissingProofs}.

\vspace{1mm}
\noindent \textbf{Why Both Components Are Needed.}
Localizing negatives confines each record's dependencies to one group, but clipping all group gradients together still leaves only one bounded contribution. Separately clipping groups that retain global negatives allows multiple contributions to accumulate, but yields the sensitivity bound of $2KC$. BGC combines local comparisons with separate group clipping, allowing multiple bounded contributions to accumulate under a sensitivity bound of $2C$. This necessity is also supported empirically by the ablation in~\Cref{fig:ablation}. How much signal actually accumulates depends on the group gradients, which we analyze next.

\vspace{1mm}
\noindent \textbf{Trade-off in the Number of Groups $K$.}
We give a theoretical analysis of how \(K\) affects the utility of \toolname{}. Let $\hat B$ denote the expected size of the Poisson-sampled batch and $S=\hat B/K$ the expected group size. Since $\hat B=K\cdot S$, choosing $K$ divides a fixed batch between the number of clipped contributions $K$ and the number of records $S$ behind each contribution. We show that $S$ and $K$ act as the traditional “batch size” of contrastive learning and DP-SGD, respectively, which creates a trade-off.

\textit{Negative-sample availability.}
Each group is a contrastive mini-batch of about $S$ records, so a record is contrasted only with the records that share its group rather than with the whole batch. In the extreme case, a record that is alone in its group has no negative at all, and its term in~\Cref{eq:simclr_group} is zero. With fewer negatives, the gradient of a group becomes a worse estimate of the gradient of the full contrastive objective $F(\theta)$, i.e., the InfoNCE objective in~\Cref{eq:simclr} computed on the whole dataset and averaged over its records~\cite{yuan2022provable}. Let $\mathcal{L}_k$ denote the inner sum of~\Cref{eq:simclr_group} for group $k$ and $\bar g_k=\nabla_\theta\mathcal{L}_k/|\mathcal{G}_k|$ its gradient per record, with $\bar g_k=0$ for an empty group. Yuan et al.~\cite{yuan2022provable} show that, under standard assumptions, the InfoNCE gradient on a mini-batch of $S$ records has a mean squared error of order $1/S$ with respect to $\nabla F$. Part of this error is a bias, because the denominator of InfoNCE is estimated from the mini-batch~\cite{DBLP:conf/nips/ChenZXCD0TZC22}. A group of BGC is such a mini-batch with a random size of expectation $S$, and Appendix~\ref{ap:group_count_analysis} proves the same order for it. If $\|\nabla_\theta s_{i,j}\|_2\le G$ for all $i$ and $j$, then
\begin{equation}
\label{eq:group_error}
\mathbb{E}\left\|\bar g_k-\nabla F(\theta)\right\|_2
\leq \gamma/\sqrt{S}=\gamma\sqrt{K/\hat B},
\end{equation}
where $\gamma=O(Ge^{2/\tau}/\tau)$ and the expectation is over the batch and the grouping. The bound grows with $K$.

\textit{Clipped gradient accumulation.}
Let $u_k=\operatorname{Clip}_C
(\nabla_\theta\mathcal{L}_k)$ denote the clipped
gradient of group $k$. The update is
\begin{equation}
\label{eq:group_count_noisy_gradient}
\widetilde g_K
=
\frac{1}{K}
\left(
\sum_{k=1}^{K}u_k+\xi
\right),
\qquad
\xi\sim\mathcal{N}
(0,4C^2\sigma^2\mathbb{I}_d).
\end{equation}
This is the update of DP-SGD on a batch of $K$ examples, in which each group acts as one example and the sensitivity is $2C$ instead of $C$. The noise in the averaged update has a root-mean-square norm of $2C\sigma\sqrt{d}/K$, where $d$ is the number of trainable parameters, so it decays as $1/K$. This is similar to how large batches improve DP-SGD~\cite{anil2022large,de2022unlocking,autoclip}. In DP-SGD, however, a larger batch also raises the sampling rate and thus requires a larger $\sigma$ for the same privacy budget, whereas increasing $K$ at a fixed $\hat B$ changes neither.
For each realized grouping, define
$\rho_K=\|\sum_k u_k\|_2/(KC)\in[0,1]$.
The ratio of the clean aggregated gradient norm
to the root-mean-square Gaussian noise norm is
\begin{equation}
\label{eq:group_count_snr}
\mathrm{SNR}_{\mathrm{RMS}}(K)
=
\frac{K\rho_K}{2\sigma\sqrt{d}},
\end{equation}
as shown in Appendix~\ref{ap:group_count_analysis}. Since $\rho_K\le1$, the SNR is at most $K/(2\sigma\sqrt{d})$. This upper limit grows linearly with $K$. How close the SNR comes to it depends on $\rho_K$, which drops when the clipped group gradients have norms below $C$ or point in different directions. Smaller groups are more prone to both.

\textit{Balancing the two effects.}
The two effects combine into one bound. As is common in DP training~\cite{autoclip,de2022unlocking}, we consider the regime in which $C$ is small enough that every nonzero group gradient is clipped, so that $u_k=C\nabla_\theta\mathcal{L}_k/\|\nabla_\theta\mathcal{L}_k\|_2$.

\begin{theorem}[Convergence of BGC]
\label{the:tradeoff}
Assume that $F$ is $L$-smooth and bounded below by $F^\star$, that every nonzero group gradient is clipped, and that $\|\nabla_\theta s_{i,j}\|_2\le G$ at every iterate. Let $F_0=F(\theta_0)-F^\star$. Then, for a suitable learning rate $\lambda$, $T$ iterations of the update $\theta_{t+1}=\theta_t-\lambda\widetilde g_K$ satisfy
\begin{equation}
\label{eq:tradeoff_bound}
\frac{1}{T}\sum_{t<T}\mathbb{E}\|\nabla F(\theta_t)\|_2
\leq
\sqrt{\frac{2LF_0}{T}\Bigl(1+\frac{4\sigma^2d}{K^2}\Bigr)}
+2\gamma\sqrt{\frac{K}{\hat B}}.
\end{equation}
\end{theorem}

\noindent The proof in Appendix~\ref{ap:group_count_analysis} combines~\Cref{eq:group_error} with the noise norm above. The first term depends on the noise through $\sigma^2d/K^2$, as does the convergence rate of DP-SGD with batch size $K$~\cite{autoclip}, and $4\sigma^2d/K^2$ is the inverse square of the SNR limit above. The second term is the residual term of contrastive learning with batch size $S$~\cite{yuan2022provable}. Increasing $K$ shrinks the first term and enlarges the second. Without noise ($\sigma=0$), the bound is smallest at $K=1$. When its minimizer $K^\star$ lies strictly inside the noise-dominated range $1\le K\le2\sigma\sqrt{d}$, $K^\star=\Theta((\sigma^2d\hat B/T)^{1/3})$, so more noise or a larger batch calls for more groups. The bound is a worst-case guarantee that ignores the partial cancellation of errors across groups, so we use it for trends rather than for the value of $K^\star$.

\textit{Toy experiment.} We illustrate the analysis on synthetic data that mimic multi-modal learning. Each positive pair $(x_i,x_i^+)$ consists of two noisy linear projections of a shared latent vector, and the latent vectors form 64 clusters, so that each record must be distinguished from similar records. Two linear encoders are trained with the update of~\Cref{the:tradeoff}, with the learning rate $\lambda$ tuned for each $K$. All values of $K$ use the same sampling rate, $\sigma$, and number of iterations, and therefore have the same privacy budget. The full setup is given in Appendix~\ref{apsec:id}. \Cref{fig:toy_k}(a) and (b) are measured at the initial parameters. \Cref{fig:toy_k}(a) shows that the group-gradient error on the left side of~\Cref{eq:group_error} stays within 4\% of the fitted curve $\hat\gamma\sqrt{K/\hat B}$ for $K\le128$, where $\hat\gamma$ is the smallest constant that keeps the curve above the error for all tested $K$. The error thus scales as $1/\sqrt{S}$ in this range. \Cref{fig:toy_k}(b) shows that $\mathrm{SNR}_{\mathrm{RMS}}(K)$ attains its bound for small $K$, where the clipped group gradients are nearly parallel, and levels off as $\rho_K$ decreases. As a result, the retrieval accuracy in~\Cref{fig:toy_k}(c) first increases and then decreases with $K$. The best $K$ shifts mildly upward with $\sigma$, from 128 for $\sigma=1$ to 160 and 192 for $\sigma=2$ and 3, in line with~\Cref{the:tradeoff}.

\begin{figure}[!t]
    \centering
    \includegraphics[width=\linewidth]{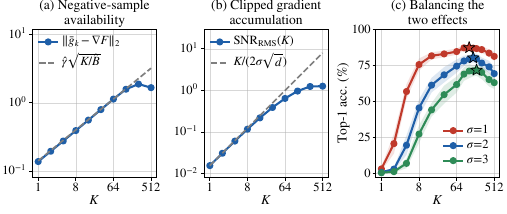}
    \caption{Toy experiment. (a) Group-gradient error on the left side of~\Cref{eq:group_error}, relative to $\|\nabla F\|_2$, and the fitted curve $\hat\gamma\sqrt{K/\hat B}$. (b) $\mathrm{SNR}_{\mathrm{RMS}}(K)$ in~\Cref{eq:group_count_snr} and its bound $K/(2\sigma\sqrt{d})$ for $\sigma=1$ and $d=1{,}024$. (c) Top-1 retrieval accuracy over five runs, where $\sigma=1,2,3$ give $\epsilon=7.9,2.6,1.6$. Stars mark the best $K$.}
    \label{fig:toy_k}
\end{figure}

We further examine this trade-off on real datasets in~\Cref{fig:tradeoff}.
The additional negative views introduced by ISA
help mitigate the loss of local contrastive
signals without increasing gradient sensitivity,
as discussed next.

\subsection{Intra-Group Sample Augmentation (ISA)} 
\label{subsec:negative_augmentation}

Group-local objectives restrict the records available as negatives. Since negative-set size and diversity affect contrastive learning~\cite{ContrastiveLearning,DBLP:conf/nips/ChenZXCD0TZC22}, ISA generates additional negative views from records already in each group. These views enrich local comparisons without introducing cross-group dependencies.

{Let $\mathcal{U}$ be a randomized augmentation intended to preserve the input's semantics. For each negative record $x_j^+$ in a group, ISA draws $N_a$ augmented views $x_j^{+,\mathrm{aug},m}=\mathcal{U}(x_j^+)$, $m=1,\ldots,N_a$, and $s_{i,j}^{\mathrm{aug},m}$ denotes the similarity in~\Cref{eq:sim} between $x_i$ and $x_j^{+,\mathrm{aug},m}$. Each view provides another negative comparison for the same record, so the loss of group $k$ and the clipped gradient sum become}
\begin{align}
{\mathcal{L}_k^{\mathrm{aug}}} &= {\sum_{i\in\mathcal{G}_k}-\log\frac{e^{s_{i,i}/\tau}}{Z_{i,k}},} \label{eq:simclr_group_aug}\\
{Z_{i,k}} &= {\sum_{j\in\mathcal{G}_k}\Bigl(e^{s_{i,j}/\tau}+\mathbf 1[j\neq i]\sum_{m=1}^{N_a} e^{s_{i,j}^{\mathrm{aug},m}/\tau}\Bigr),} \nonumber\\
{g_{\mathrm{group}}^{\mathrm{aug}}} &= {\sum_{k=1}^{K}\operatorname{Clip}_C\bigl(\nabla_\theta\mathcal{L}_k^{\mathrm{aug}}\bigr).} \label{eq:f_group_aug}
\end{align}
Because every augmented view comes from a record in the same group, the augmented gradient sum retains the sensitivity bound of $2C$. We provide the proof of \Cref{the:group_grad_aug} in Appendix~\ref{apsec:MissingProofs}.

\begin{theorem}
\label{the:group_grad_aug}
Under add-or-remove adjacency, the unnormalized augmented estimator $g_{\mathrm{group}}^{\mathrm{aug}}$ in \Cref{eq:f_group_aug} has $\ell_2$ sensitivity bounded by \(\Delta_{g_{\mathrm{group}}^{\mathrm{aug}}} \leq 2C\), for any number of augmented views $N_a\geq0$.
\end{theorem}

\begin{algorithm}[!t]
       \caption{\toolname{} Workflow}
      \label{alg:method}
       \SetKwInOut{Input}{Input}
      \SetKwInOut{Output}{Output}
      \SetKwProg{Fn}{Function}{:}{}
    \Input{Encoder $\mathcal{E}$ parameterized with $\theta$, sensitive dataset $\mathcal{D}$, group number $K$, augmentation $\mathcal{U}$, number of augmented views per negative record $N_a$, training iterations $T$, learning rate $\lambda$, Poisson sampling rate $q$, clipping norm $C$, noise scale $\sigma$.}
    \Output{Trained encoder $\mathcal{E}$}
    $t = 0$\\
    \While{$t < T$}{
    Draw $\mathcal{B}=\{(x_i,x_i^+)\}_{i=1}^B$ from $\mathcal{D}$ by Poisson sampling with rate $q$\\
    Draw $r_i\sim\mathrm{Unif}\{1,\ldots,K\}$ independently and set $\mathcal{G}_k=\{i: r_i=k\}$, $k=1,\ldots,K$\\
    \tcp{Group-level InfoNCE Loss}
    \For{$k \gets 1, \ldots, K$ }{
        Compute similarity $s$ (Eq.~\ref{eq:sim})\\%
        Compute similarity $\{s^{\mathrm{aug},m}\}_{m=1}^{N_a}$ {(Eq.~\ref{eq:sim})}\\
        {Compute $\mathcal{L}_k^{\mathrm{aug}}$ (Eq.~\ref{eq:simclr_group_aug})}\label{line:loss}
        }
    \tcp{Group Clipping}
    $\Delta\theta = \sum\limits_{k=1}^{K} { \text{Clip}_C\left(\nabla_\theta \mathcal{L}_k^{{\mathrm{aug}}}\right)}$\\
    $\Delta\tilde{\theta} = \frac{1}{K}(\Delta\theta + \xi)$, $\xi \sim \mathcal{N}(0, 4C^2\sigma^2 \mathbb{I})$\\
    $\theta = \theta - \lambda \Delta\tilde{\theta}$\label{line:dpsgd}\\
    $t = t + 1$\\
    }
     \KwRet: $\mathcal{E}$
\end{algorithm}

\noindent \Cref{alg:method} elaborates on the workflow of \toolname{}. At each iteration, \toolname{} first samples a mini-batch $\mathcal{B}$ from the original dataset $\mathcal{D}$. It then independently assigns each sample pair to one of $K$ groups. The resulting group sizes have expectation $B/K$. Next, it computes the augmented group-local InfoNCE loss in~\Cref{eq:simclr_group_aug} for each group. Gradients from each group are clipped individually, aggregated, and perturbed with Gaussian noise for the DP guarantee. Finally, the model parameters are updated using the noisy gradient.

\begin{table*}[!t]
\small
    \centering
    \caption{Summary of investigated datasets used for evaluation.}
    \label{tab:datainfo}
    \setlength{\tabcolsep}{2.9mm}{
    \resizebox{1.0\textwidth}{!}{
    \begin{tabular}{ll*{7}{c}}
    \toprule
    \textbf{Dataset} & \textbf{Size} & \textbf{Image Dimension} & \textbf{Text Length} & \textbf{Category}  & \textbf{Training} & \textbf{Validation} & \textbf{Test} & \textbf{Task}  \\
    \midrule
    {FashionMNIST}~\cite{fmnist} & 70,000 & 28$\times$28$\times$1 & - & 10 & 55,000 & 5,000 & 10,000 &  \\
    {CIFAR-10}~\cite{cifar10} & 60,000 & 32$\times$32$\times$3 & - & 10 & 45,000 & 5,000 & 10,000 & Classification \\
    {EuroSAT}~\cite{eurosat} & 27,000 & 64$\times$64$\times$3 & - & 10 & 21,000 & 2,000 & 4,000 & (uni-modal) \\
    {Camelyon}~\cite{camelyon1} & 337,340 & 64$\times$64$\times$3 & - & 2 & 269,538 & 32,898  & 34,904 &  \\
    {ImageNet}~\cite{imagenet} & 1,431,167 & 256$\times$256$\times$3 & - & 1000 & 1,281,167 & 50,000& 100,000 &  \\
    \midrule
    {PEDES}~\cite{PEDES} & 238,768  & 128$\times$384$\times$3 & 13$\sim$96 & - & 190,916 & 23,882 & 23,970 &   \\
    {RSTPReid}~\cite{RSTPReid} & 41,010  & 200$\times$682$\times$3 & 11$\sim$469 & - & 37,010 & 2,000 & 2,000 &  Retrieval \\
    {Fashion}~\cite{Fashion} & 42,537  & 153$\times$224$\times$3 & 4$\sim$93 & - & 32,029 & 2,000 & 8,508 &  (multi-modal) \\
    {ROCO}~\cite{ROCO} & 79,793  & 682$\times$784$\times$3 & 1$\sim$778 & - & 59,962 & 9,904 & 9,927 &   \\
    \bottomrule
\end{tabular}
}}
\end{table*}

\subsection{Privacy Analysis}
\label{subsec:dp_analysis}

We analyze the privacy of \toolname{} for the group-local estimator $g_{\mathrm{group}}$ in~\Cref{eq:f_group}, and the extension to ISA is discussed at the end of this section. \Cref{the:group_grad} bounds the $\ell_2$ sensitivity of the unnormalized clipped gradient sum by $2C$. Adding $\mathcal{N}(0,4C^2\sigma^2\mathbb{I})$ therefore corresponds to a noise scale of $\sigma$ relative to this sensitivity, and the remaining step is the standard analysis of the Poisson-subsampled Gaussian mechanism~\cite{sgm} composed over $T$ iterations~\cite{rdp}.

\noindent\textbf{Sampled Gaussian Mechanism.}
For $\alpha>1$, let $\mu_q=(1-q)\,\mathcal{N}(0,\sigma^2)+q\,\mathcal{N}(1,\sigma^2)$ and define
\begin{equation}
\label{eq:eps_sgm}
\begin{aligned}
\varepsilon_{\mathrm{SGM}}(\alpha;q,\sigma)
=\max\bigl\{&
D_\alpha\bigl(\mathcal{N}(0,\sigma^2)\,\|\,\mu_q\bigr),\\
&D_\alpha\bigl(\mu_q\,\|\,\mathcal{N}(0,\sigma^2)\bigr)
\bigr\},
\end{aligned}
\end{equation}
where $D_\alpha$ is the R\'enyi divergence of order $\alpha$~\cite{rdp}. This is the order-$\alpha$ RDP bound of the Poisson-sampled Gaussian mechanism with unit $\ell_2$ sensitivity~\cite{sgm}, and it is evaluated numerically by standard accountants~\cite{sgm,opacus}.

\begin{theorem}[End-to-End Privacy of \toolname{}]
\label{the:end_to_end_privacy}
Let \toolname{} run for $T$ iterations with Poisson sampling rate $q$, fresh independent group assignment at every iteration, and Gaussian noise $\mathcal{N}(0,4C^2\sigma^2\mathbb{I})$ added to the sum of clipped group gradients $g_{\mathrm{group}}$ in~\Cref{eq:f_group}. Then the sequence of noisy updates, and hence the final encoder, satisfies, for every $\alpha>1$,
\begin{equation}
\label{eq:end_to_end_rdp}
\left(\alpha,\,T\varepsilon_{\mathrm{SGM}}(\alpha;q,\sigma)\right)
\text{-RDP}.
\end{equation}
Consequently, for any $\delta\in(0,1)$, the encoder satisfies $(\epsilon_\delta,\delta)$-DP with
\begin{equation}
\label{eq:end_to_end_dp}
\begin{aligned}
\epsilon_\delta
=\inf_{\alpha>1}\Bigl\{&
T\varepsilon_{\mathrm{SGM}}(\alpha;q,\sigma)
+\frac{\log(1/\delta)-\log\alpha}{\alpha-1}\\
&+\log\frac{\alpha-1}{\alpha}
\Bigr\}.
\end{aligned}
\end{equation}
\end{theorem}
The proof is given in Appendix~\ref{apsec:MissingProofs}. \Cref{the:end_to_end_privacy} shows that the privacy accounting of \toolname{} is identical to that of DP-SGD with the same $(q,\sigma,T)$, and the group count $K$ affects utility but not the privacy budget. In practice, we calibrate $\sigma$ for a target $(\epsilon,\delta)$ with the RDP accountant of {Opacus~\cite{opacus}, which evaluates~\Cref{eq:end_to_end_dp} over a grid of orders $\alpha$}. As in DP-SGD, the encoders contain no cross-sample normalization, and the ResNet18 uses GroupNorm in place of BatchNorm~\cite{opacus}.

\noindent\textbf{Extension to ISA.}
With ISA, the clipped sum is $g_{\mathrm{group}}^{\mathrm{aug}}$ in~\Cref{eq:f_group_aug}, whose sensitivity is also bounded by $2C$ (\Cref{the:group_grad_aug}). Each augmented view is a function of a single record in the same group and of augmentation randomness that is drawn independently of the data, so the same argument applies with the augmentation randomness treated like the group marks, and \Cref{the:end_to_end_privacy} holds for \Cref{alg:method} with any $N_a\geq0$.

\noindent\textbf{Tighter Accounting.}
Because each step reduces to the sampled Gaussian mechanism, the same guarantee can also be computed with privacy loss distribution (PLD) accountants, including Fourier~\cite{koskela2021tight} and PRV~\cite{prv} accountants, which avoid the RDP-to-DP conversion and typically yield a smaller $\epsilon$ for the same $(q,\sigma,T)$. All reported results use the RDP accountant.

\section{Experimental Setup}
\label{sec:exp_setup}

\noindent \textbf{Evaluation Tasks.} We evaluate \toolname{} on both \textit{uni-modal} and \textit{multi-modal} settings to assess different aspects of the learned representations of contrastive learning.

\begin{table*}[!t]
\renewcommand{\arraystretch}{1.1}
\setlength{\tabcolsep}{5.5pt}
    \centering
    \caption{Main results under privacy budgets $\epsilon\in\{1,10\}$. (a) Uni-modal classification accuracy (\%) of linear probing and $k$-NN ($k=3$). (b) Multi-modal top-10 retrieval accuracy (\%) of image-to-text (I$\rightarrow$T) and text-to-image (T$\rightarrow$I) retrieval.}
    \label{tab:main_cls}
    \label{tab:main_retrieval}
    \resizebox{1.0\textwidth}{!}{
    \begin{tabular}{l*{16}{c}}
    \toprule
    & \multicolumn{8}{c}{$\epsilon=1$} & \multicolumn{8}{c}{$\epsilon=10$}\\
    \cmidrule(lr){2-9}\cmidrule(lr){10-17}
    \multicolumn{17}{l}{\textit{(a) Uni-modal: classification accuracy (\%)}}\\
    \addlinespace[2pt]
    \multirow{2}{*}{\textbf{Method}} & \multicolumn{2}{c}{F-MNIST} & \multicolumn{2}{c}{CIFAR-10} & \multicolumn{2}{c}{EuroSAT} & \multicolumn{2}{c}{Camelyon} & \multicolumn{2}{c}{F-MNIST} & \multicolumn{2}{c}{CIFAR-10} & \multicolumn{2}{c}{EuroSAT} & \multicolumn{2}{c}{Camelyon} \\
    \cmidrule(lr){2-3}\cmidrule(lr){4-5}\cmidrule(lr){6-7}\cmidrule(lr){8-9}\cmidrule(lr){10-11}\cmidrule(lr){12-13}\cmidrule(lr){14-15}\cmidrule(lr){16-17}
     & Linear & $k$-NN & Linear & $k$-NN & Linear & $k$-NN & Linear & $k$-NN & Linear & $k$-NN & Linear & $k$-NN & Linear & $k$-NN & Linear & $k$-NN \\
    \midrule
    Base & 76.0 & 77.4 & 30.6 & 24.5 & 50.5 & 33.4 & 69.7 & 52.4 & 76.0 & 77.4 & 30.6 & 24.5 & 50.5 & 33.4 & 69.7 & 52.4 \\
    DP-SGD~\cite{dpsgd} & 77.5 & 79.4 & 32.0 & 25.6 & 48.2 & 35.6 & 69.6 & 54.8 & 77.5 & 78.8 & 33.6 & 28.5 & 50.1 & 36.1 & 68.8 & 54.3 \\
    Logit-DP~\cite{pair-level} & 78.2 & 79.6 & 31.8 & 25.4 & 48.2 & 34.5 & 68.3 & 57.6 & 78.4 & 79.7 & 32.1 & 25.3 & 51.4 & 34.3 & 65.7 & 54.6 \\
    DP-CLIP~\cite{batch-level} & 78.3 & 80.0 & 33.4 & 28.1 & 50.6 & 36.5 & 72.0 & 66.8 & 79.8 & 81.1 & 35.1 & 28.9 & 52.3 & 38.2 & 73.5 & 69.9 \\
    \midrule
    \rowcolor{gray!12} \toolname{} (ours) & 81.5 & 82.1 & 37.6 & 33.8 & 52.3 & 40.2 & 77.5 & 73.5 & 82.9 & 82.8 & 42.3 & 39.5 & 59.1 & 53.3 & 78.4 & 76.7 \\
    \midrule[\heavyrulewidth]
    \multicolumn{17}{l}{\textit{(b) Multi-modal: top-10 retrieval accuracy (\%)}}\\
    \addlinespace[2pt]
    \multirow{2}{*}{\textbf{Method}} & \multicolumn{2}{c}{PEDES} & \multicolumn{2}{c}{RSTPReid} & \multicolumn{2}{c}{Fashion} & \multicolumn{2}{c}{ROCO} & \multicolumn{2}{c}{PEDES} & \multicolumn{2}{c}{RSTPReid} & \multicolumn{2}{c}{Fashion} & \multicolumn{2}{c}{ROCO} \\
    \cmidrule(lr){2-3}\cmidrule(lr){4-5}\cmidrule(lr){6-7}\cmidrule(lr){8-9}\cmidrule(lr){10-11}\cmidrule(lr){12-13}\cmidrule(lr){14-15}\cmidrule(lr){16-17}
     & I$\rightarrow$T & T$\rightarrow$I & I$\rightarrow$T & T$\rightarrow$I & I$\rightarrow$T & T$\rightarrow$I & I$\rightarrow$T & T$\rightarrow$I & I$\rightarrow$T & T$\rightarrow$I & I$\rightarrow$T & T$\rightarrow$I & I$\rightarrow$T & T$\rightarrow$I & I$\rightarrow$T & T$\rightarrow$I \\
    \midrule
    Base & 15.5 & 10.7 & 15.3 & 11.7 & 3.5 & 2.4 & 20.5 & 20.2 & 15.5 & 10.7 & 15.3 & 11.7 & 3.5 & 2.4 & 20.5 & 20.2 \\
    DP-SGD~\cite{dpsgd} & 14.5 & 10.4 & 13.4 & 11.0 & 2.9 & 2.3 & 18.7 & 17.9 & 14.2 & 10.7 & 14.2 & 11.1 & 3.5 & 2.5 & 20.4 & 18.8 \\
    Logit-DP~\cite{pair-level} & 16.1 & 11.7 & 15.4 & 12.6 & 3.4 & 2.7 & 18.7 & 18.2 & 16.4 & 12.4 & 16.6 & 12.6 & 3.5 & 3.1 & 19.7 & 19.5 \\
    DP-CLIP~\cite{batch-level} & 15.1 & 11.7 & 15.4 & 11.8 & 3.1 & 2.4 & 18.4 & 18.0 & 18.5 & 13.3 & 17.2 & 14.4 & 3.6 & 3.0 & 19.4 & 20.5 \\
    \midrule
    \rowcolor{gray!12} \toolname{} (ours) & 55.8 & 40.0 & 25.9 & 21.4 & 11.9 & 9.5 & 23.2 & 25.3 & 65.9 & 48.1 & 38.0 & 34.5 & 26.3 & 24.7 & 37.5 & 38.4 \\
    \bottomrule
\end{tabular}
}
\end{table*}

\begin{itemize}[leftmargin=*]
    \item \textit{Uni-modal evaluation (image-only)}: We train an image encoder on image datasets through DP contrastive learning. We use two standard protocols to probe the quality of image features extracted by the trained encoder: (1) Linear probing, where we freeze the weights of trained image encoder and train a linear classifier that maps its fixed visual features to class labels for image classification, and (2) $k$-NN classification, which directly predicts labels by comparing test features to training features using cosine similarity. 
    \item \textit{Multi-modal evaluation (image–text)}: We train an image encoder and text encoder on image-text datasets through DP contrastive learning, and assess their performance on image-to-text and text-to-image retrieval. For instance, given a sentence, the model retrieves the most relevant image from a dataset by comparing feature similarity, and vice versa.
\end{itemize}

These evaluation protocols are widely used and well-established in previous contrastive learning~\cite{CLIP,simclr}, providing reliable and comparable measures of performance.

\noindent \textbf{Investigated Datasets.} We evaluate four uni-modal datasets: Fashion-MNIST (F-MNIST)~\cite{fmnist}, CIFAR-10~\cite{cifar10}, EuroSAT~\cite{eurosat}, and Camelyon~\cite{camelyon1}, and four multi-modal datasets: CUHK-PEDES~\cite{PEDES}, RSTPReid~\cite{RSTPReid}, Fashion~\cite{Fashion}, and ROCO~\cite{ROCO}. Dataset statistics and their data splits are in~\Cref{tab:datainfo}.

\begin{table}[!t]
\renewcommand{\arraystretch}{1.1}
\setlength{\tabcolsep}{3.5pt}
    \centering
    \caption{Classification accuracy (\%) of {linear probing} and $k$-NN ($k=3$), where ImageNet is used for pre-training.}
    \label{tab:main_cls_transfer}
    \resizebox{0.48\textwidth}{!}{
    \begin{tabular}{l*{8}{c}}
    \toprule
    & \multicolumn{4}{c}{$\epsilon=1$} & \multicolumn{4}{c}{$\epsilon=10$}\\
    \cmidrule(lr){2-5}\cmidrule(lr){6-9}
    \multirow{2}{*}{\textbf{Method}} & \multicolumn{2}{c}{F-MNIST} & \multicolumn{2}{c}{CIFAR-10} & \multicolumn{2}{c}{F-MNIST} & \multicolumn{2}{c}{CIFAR-10} \\
    \cmidrule(lr){2-3}\cmidrule(lr){4-5}\cmidrule(lr){6-7}\cmidrule(lr){8-9}
     & Linear & $k$-NN & Linear & $k$-NN & Linear & $k$-NN & Linear & $k$-NN \\
    \midrule
    Base & 76.0 & 77.4 & 30.6 & 24.5 & 76.0 & 77.4 & 30.6 & 24.5 \\
    DP-SGD~\cite{dpsgd} & 76.9 & 77.7 & 31.1 & 25.5 & 76.7 & 78.1 & 32.2 & 26.4\\
    Logit-DP~\cite{pair-level} & 77.0 & 78.5 & 31.3 & 26.6 & 79.5 & 81.5 & 34.0 & 30.0\\
    DP-CLIP~\cite{batch-level} & 76.7 & 78.0 & 31.1 & 26.6 & 77.1 & 77.7 & 31.2 & 26.9 \\
    \midrule
    \rowcolor{gray!12} \toolname{} (ours) & 79.5 & 81.0 & 34.5 & 32.1 & 82.3 & 82.9 & 40.2 & 38.2 \\
    \bottomrule
\end{tabular}
}
\end{table}

\vspace{0.5mm}
\noindent \textbf{Baselines.} We compare \toolname{} with three representative methods: DP-SGD~\cite{dpsgd} and DP-CLIP~\cite{batch-level}, which correspond to the sample-level and batch-level bounding in~\Cref{subsec:dpcl}, and Logit-DP~\cite{pair-level}, which clips pairwise similarity gradients before aggregation. The embedding models are a randomly initialized ResNet18 and a pre-trained ViT\footnote{\url{https://huggingface.co/openai/clip-vit-base-patch32}} for uni-modal and multi-modal, respectively. Additionally, to demonstrate the effectiveness of DP contrastive learning, we include a baseline named ``Base,'' where the embedding model is not trained on the private dataset and is instead used directly for the downstream task.

\vspace{0.5mm}
\noindent \textbf{Implementation.} We use the same DP optimizer, privacy accountant, model initialization, and training protocol for all methods. For \toolname{}, we set the expected group size to $S=16$, i.e., $K=512$ with an expected batch size of 8,192 for multi-modal tasks and $K=128$ with 2,048 for uni-modal tasks, and use $N_a=1$ augmented view per negative record. All hyper-parameters are listed in~\Cref{tab:hp_ours} in Appendix~\ref{apsec:id}.

\vspace{0.5mm}
\noindent \textbf{Evaluation Metrics.} For uni-modal evaluation, we use classification accuracy. For the multi-modal task, we use top-$h$ retrieval accuracy. Specifically, given a query from one modality (e.g., an image), we rank all items in the other modality (e.g., texts) by their embedding similarity and check whether the ground-truth match is within the top-$h$ retrieved results. We evaluate both image-to-text and text-to-image directions, and set $h$ to 10. 
For the baselines on multi-modal tasks, whose training is unstable, we report their best test accuracy observed during training.

\section{Results Analysis}
\label{sec:results_analysis}

We evaluate \toolname{} in terms of overall utility, ablation study, generalization across training strategies, and hyper-parameter analysis. We compare all methods under the same privacy budgets and focus on the utility of contrastive learning. Since DP provides a formal privacy guarantee, we compare utility under matched privacy budgets and treat empirical privacy attacks as complementary evidence rather than the primary metric for method comparison~\cite{pair-level,sinha2025vaultgemma,privimage,dp_classification1,dp-feta,jayaraman2019evaluating,nasr2021adversary,nasr2023tight,aerni2024evaluations}. Nevertheless, we report a membership inference attack (MIA) evaluation in Appendix~\ref{subsec:mia} as an additional empirical check. We compare \toolname{} with the non-DP baseline in Appendix~\ref{subsec:non-private}.

\begin{figure*}[!t]
    \centering
    \begin{subfigure}[b]{0.48\textwidth}
        \centering
        \includegraphics[width=\textwidth]{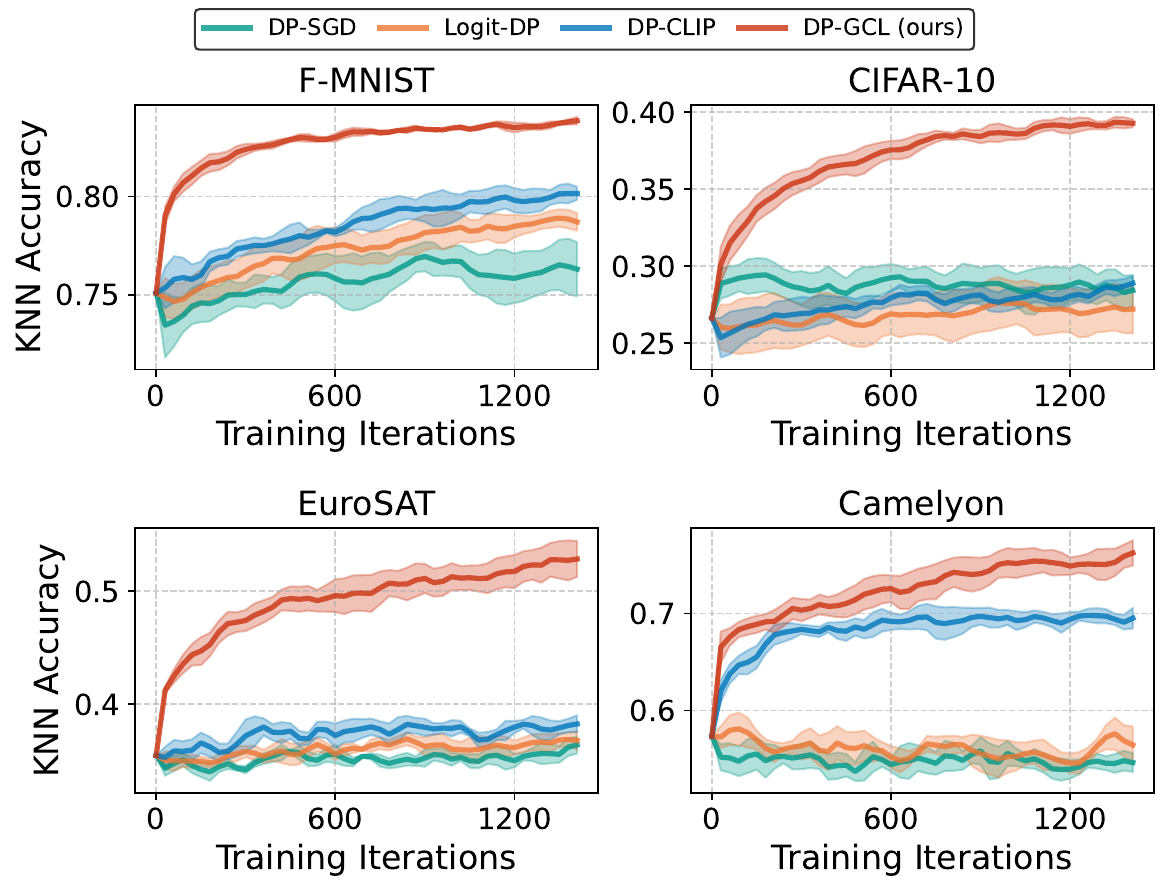}
        \caption{Uni-modal Evaluation}
        \label{fig:acc_change_uni}
    \end{subfigure}%
    \hfill
    \begin{subfigure}[b]{0.48\textwidth}
        \centering
        \includegraphics[width=\textwidth]{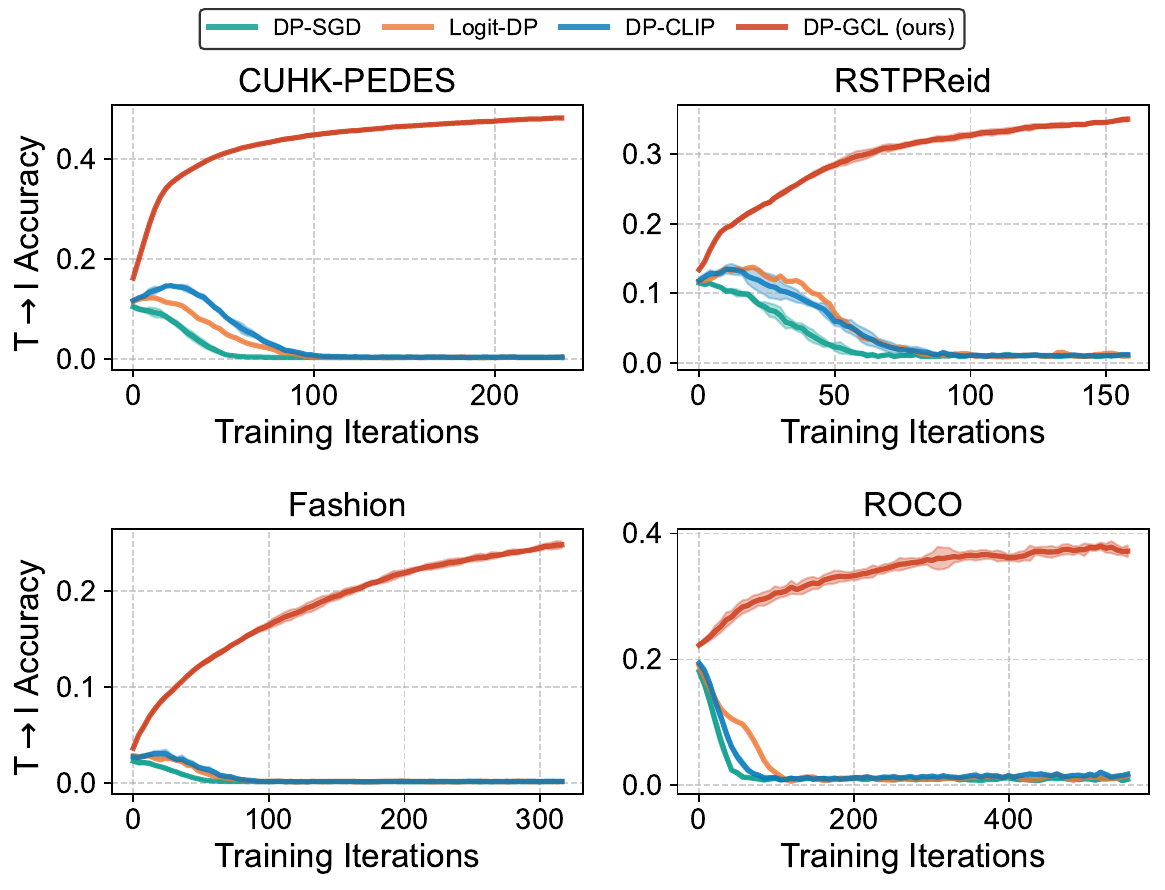}
        \caption{Multi-modal Evaluation}
        \label{fig:acc_change_multi}
    \end{subfigure}
    \caption{(a) \(k\)-NN classification accuracy and (b) T$\rightarrow$I retrieval accuracy of three baselines and \toolname{} during training.}
    \label{fig:acc_change}
\end{figure*}

\subsection{Overall Performance}
\label{subsec:overall_results}

\subsubsection{Utility of Features} For uni-modal evaluation,~\Cref{tab:main_cls}{(a)} reports linear probing and \(k\)-NN (\(k\)=3) accuracy across four datasets under privacy budgets $\epsilon\in\{1,10\}$. The “Base” baseline (without contrastive learning) performs poorly, especially on complex datasets like {Camelyon} (52.4\% \(k\)-NN accuracy). DP baselines stay close to this baseline, whereas \toolname{} achieves the best results in all settings, with the largest gain on {EuroSAT} (+15.1 percentage points in \(k\)-NN accuracy over DP-CLIP at $\epsilon=10$). Performance improves with higher privacy budgets. \toolname{} also narrows the gap between \(k\)-NN and linear probing, which indicates that its features separate the classes even without a trained classifier.

We further assess the generalization capability of \toolname{} under transfer learning settings. Specifically, we train the image encoder on ImageNet~\cite{imagenet} via DP contrastive learning, and evaluate the encoder on other datasets. In this setting, the training images and test images could have totally different data distributions. Table~\ref{tab:main_cls_transfer} reports the classification accuracy on downstream tasks (F-MNIST and CIFAR-10) using ImageNet~\cite{imagenet} pre-trained representations. The results on EuroSAT and Camelyon are provided in~\Cref{tab:main_cls_transfer_ap}. This setup tests the robustness of the learned features when transferred to distinct domains under privacy constraints. \toolname{} leads all reported metrics on F-MNIST and CIFAR-10, and most settings on EuroSAT and Camelyon. One exception is Camelyon linear probing at $\epsilon=1$.

For multi-modal evaluation,~\Cref{tab:main_retrieval}{(b)} presents image-to-text (I→T) and text-to-image (T→I) retrieval accuracy across four datasets under privacy budgets $\epsilon\in\{1,10\}$. The “Base” baseline achieves modest performance, %
while standard DP adaptations like DP-SGD often degrade retrieval accuracy, likely due to excessive noise disrupting alignment between modalities. Among existing DP methods, Logit-DP and DP-CLIP show marginal improvements, but their performance remains limited, especially on fine-grained datasets like {Fashion} ($\leq$3.6\% I→T at $\epsilon=10$). 
\toolname{} improves on them in all settings. At $\epsilon=1$, \toolname{} improves I→T accuracy over DP-CLIP by 40.7 percentage points on {PEDES} and 10.5 percentage points on {RSTPReid}.
At $\epsilon=10$, it reaches 65.9\% I→T on {PEDES} and 26.3\% on Fashion, whereas no baseline exceeds 18.5\% and 3.6\%, respectively.

\begin{figure*}
    \centering
    \includegraphics[width=1.0\linewidth]{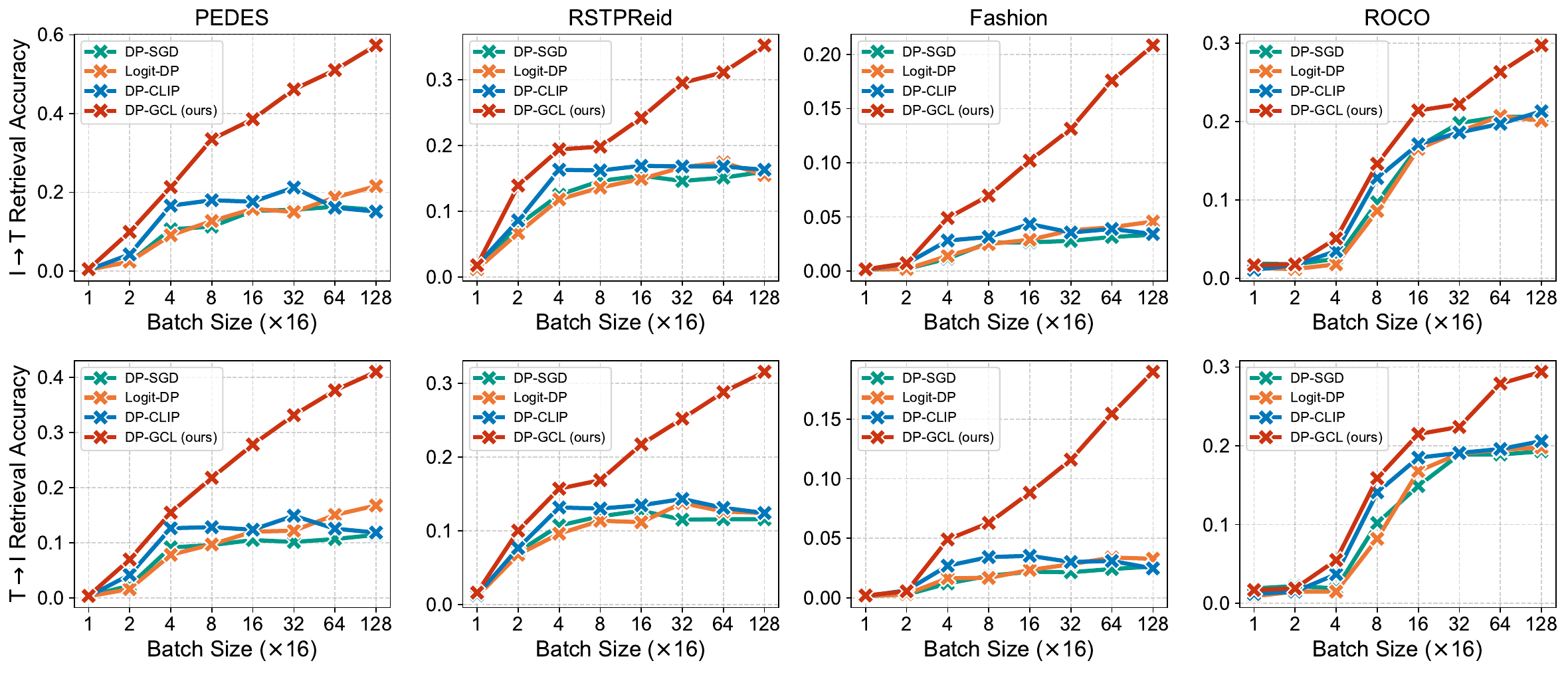}
    \caption{Image-to-text (I$\rightarrow$T) and text-to-image (T$\rightarrow$I) retrieval accuracy of three baselines and \toolname{} under different batch sizes $\{1,2,4,8,16,32,64,128\} \times 16$.}
    \label{fig:batch_size_scaling}
\end{figure*}

\subsubsection{Convergence Efficiency}~\Cref{fig:acc_change} shows the accuracy on test sets over training iterations for (a) uni-modal and (b) multi-modal tasks under $\epsilon=10$. We also provide the curve of training loss in Appendix~\ref{ap:subsec:convergence}.

\Cref{fig:acc_change} (a) shows that \toolname{} reaches higher accuracy and converges faster than the baselines. On {EuroSAT} and {Camelyon}, the \(k\)-NN accuracy of DP-SGD and Logit-DP barely improves during training, whereas that of \toolname{} increases steadily. In the multi-modal task (\Cref{fig:acc_change} (b)), the T→I accuracy of DP-SGD, Logit-DP, and DP-CLIP collapses to nearly zero during training, whereas that of \toolname{} keeps increasing on all four datasets.

Retrieval requires finer distinctions between positive and negative pairs than classification, and the multi-modal task separates the methods more clearly. The remaining experiments therefore focus on it.

\subsubsection{Batch Size Scaling}~\Cref{fig:batch_size_scaling} illustrates how retrieval accuracy (I→T and T→I) evolves with increasing batch size (×16 scale) across four datasets under $\epsilon=10$. We fixed the expected group size as 16 to set the group number \(K = \lceil \hat{B}/S\rceil\), where \(\hat{B}\) is the expected batch size. The baselines plateau or even degrade at larger batch sizes, especially on {Fashion} and {PEDES}, whereas the accuracy of \toolname{} keeps increasing with the batch size.

\subsubsection{Consumption of Computational Resource}
We also compare the peak GPU memory usage and runtime of different baselines and \toolname{} on CIFAR-10 and PEDES in~\Cref{tab:computationalResource}. Logit-DP incurs substantially higher runtime, exceeding 60 hours for uni-modal and 72 hours for multi-modal tasks, primarily due to its \(\mathcal{O}(B^2)\) time complexity. Following its original implementation~\cite{pair-level}, we use small batches with gradient accumulation to ensure computational feasibility. Although this strategy reduces peak memory usage, it does not alleviate the quadratic computational bottleneck. In contrast, \toolname{} achieves substantial utility improvements with only a marginal computational overhead over other baselines, i.e., DP-SGD and DP-CLIP. 

\begin{table}[!t]
\renewcommand{\arraystretch}{1.1}
\setlength{\tabcolsep}{5pt}
    \centering
    \caption{{Peak GPU memory (GB) and training time (hours) of the baselines and \toolname{} on CIFAR-10 (uni-modal) and PEDES (multi-modal).}}
    \label{tab:computationalResource}
    \resizebox{0.48\textwidth}{!}{
    \begin{tabular}{lcccc}
    \toprule
    \multirow{2}{*}{\textbf{Method}} & \multicolumn{2}{c}{Peak memory (GB)} & \multicolumn{2}{c}{Time (h)}\\
    \cmidrule(lr){2-3}\cmidrule(lr){4-5}
     & Uni-modal & Multi-modal & Uni-modal & Multi-modal \\
    \midrule
    DP-SGD~\cite{dpsgd} & 32.6 & 22.5 & 0.80 & 2.43 \\
    Logit-DP~\cite{pair-level} & 1.8 & 2.1 & 60.44 & 72.18 \\
    DP-CLIP~\cite{batch-level} & 32.6 & 22.5 & 0.80 & 2.43 \\
    \midrule
    \rowcolor{gray!12} \toolname{} (ours) & 33.2 & 22.6 & 0.99 & 2.89 \\
    \bottomrule
\end{tabular}
}
\end{table}

\begin{figure}
    \centering
    \includegraphics[width=1.0\linewidth]{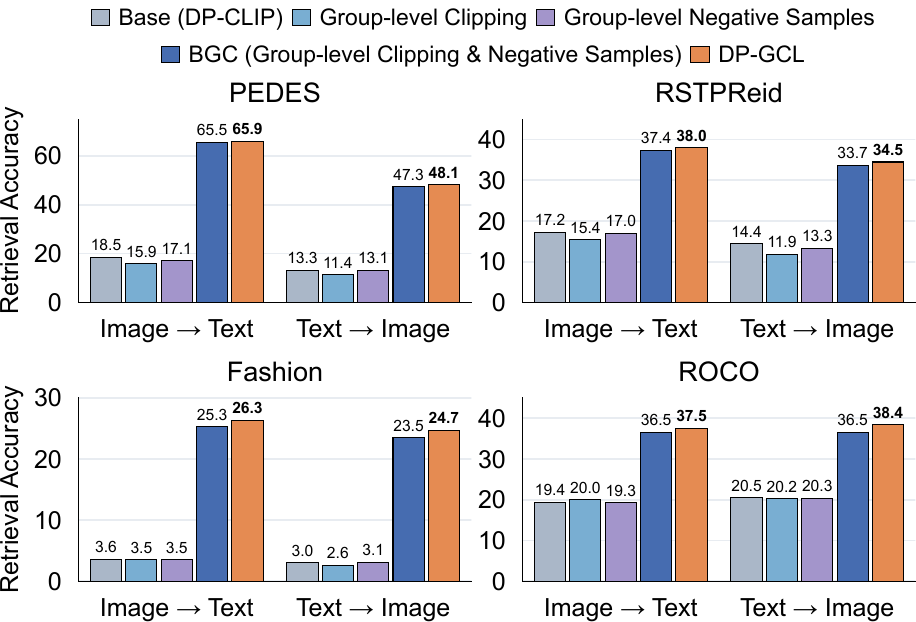}
    \caption{Ablation study of bounding group-level contribution (BGC), including group-level clipping and {group-local contrastive objective}, and intra-group sample augmentation (ISA). \toolname{} is BGC + ISA.}
    \label{fig:ablation}
\end{figure}

\subsection{Ablation Study}
\label{subsec:ablation}

\begin{figure}[!t]
    \centering
    \setlength{\abovecaptionskip}{0pt}
    \includegraphics[width=0.98\linewidth]{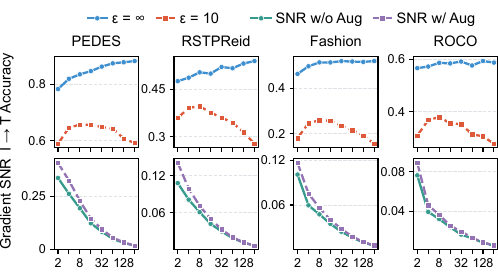}
    \caption{I→T retrieval accuracy (first row) and {gradient signal-to-noise ratio (SNR, second row)} with different {expected} group sizes.}
    \label{fig:tradeoff}
\end{figure}

This section investigates how the two key modules of \toolname{}, i.e., bounding group-level contribution (BGC) and intra-group sample augmentation (ISA), improve the privacy-utility trade-off in DP contrastive learning. To further understand the effectiveness of BGC, we disentangle it into two components: group-level clipping and {group-local contrastive objective}, and evaluate their individual and combined contributions.

As shown in \Cref{fig:ablation}, compared to DP-CLIP (batch-level clipping), applying either component of BGC alone yields limited or even degraded performance. In contrast, combining both components substantially improves I$\rightarrow$T accuracy on PEDES to 65.5\%. Similar trends are observed across all four datasets, demonstrating that the performance gains primarily arise from their joint design rather than either component in isolation. These results highlight the importance of aligning the contrastive dependency structure with the gradient clipping units. Group-level clipping alone cannot eliminate cross-group dependencies, while {group-local contrastive objective} alone still suffers from the limited gradient accumulation imposed by batch-level clipping. By jointly localizing contrastive interactions and clipping gradients at the group level, BGC achieves low sensitivity while preserving effective gradient accumulation, leading to substantial utility improvements. Adding ISA improves retrieval on all four datasets, although the gains are modest with $N_a=1$. We further demonstrate that increasing the number of augmentations can provide additional improvements in \Cref{subsec:hyperparameters}.

To further explore the reason for BGC's significant improvement, we calculate (1) the I→T retrieval accuracy under $\epsilon=\{10,\infty\}$, where $\epsilon=\infty$ means non-private, and (2) the gradient Signal-to-Noise Ratio (SNR) of \toolname{} with and without intra-group sample augmentation under different expected group size $S$ on {PEDES} and {RSTPReid}. The gradient SNR is defined as the norm of clean aggregated gradients divided by the norm of Gaussian noise at the first training iteration. It measures the size of the aggregated signal relative to the noise, not its direction. Although \toolname{} directly determines grouping through the number of groups \(K\), we use the expected group size \(S\) to parameterize the grouping granularity in this analysis, as \(S\) provides a more intuitive characterization of the trade-off between negative-sample diversity and gradient SNR.

\begin{figure}[!t]
    \centering
    \setlength{\abovecaptionskip}{0pt}
    \includegraphics[width=0.98\linewidth]{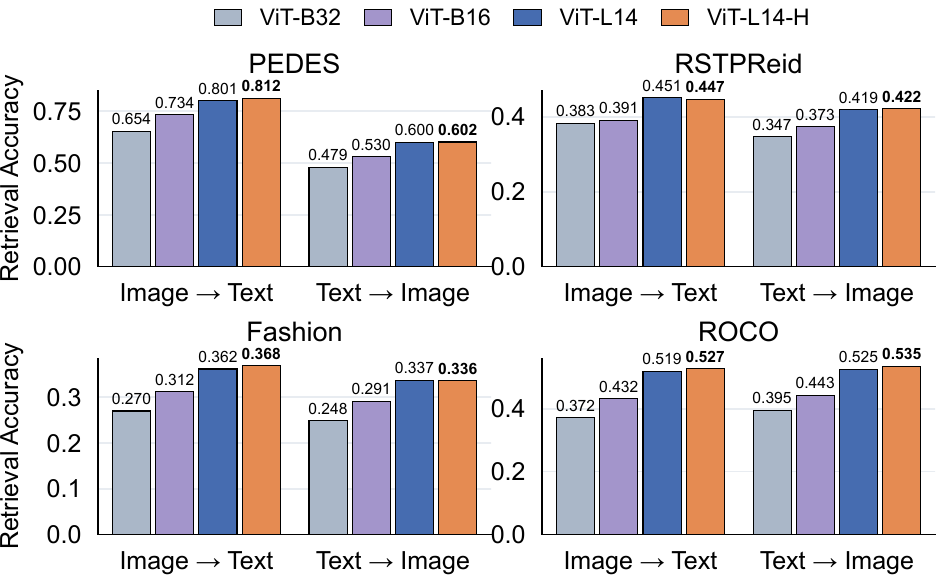}
    \caption{Retrieval accuracy of \toolname{} with different network architectures for the multi-modal encoders.}
    \label{fig:model_arch}
\end{figure}

\Cref{fig:tradeoff} shows the overall results. Under non-private settings ($\epsilon=\infty$), the I→T retrieval accuracy generally improves as group size increases (i.e., group number \(K\) decreases) across both {PEDES} and {RSTPReid} (first row in~\Cref{fig:tradeoff}, blue curve), as~\Cref{the:tradeoff} suggests for $\sigma=0$. Larger groups enhance contrastive learning by providing richer negative samples, leading to better cross-modal alignment. This finding is consistent with numerous prior non-private contrastive learning studies (e.g., SimCLR~\cite{simclr}, CLIP~\cite{CLIP}), which similarly observe improved representation quality with more negative samples. The performance stabilizes or slightly plateaus at large group sizes (e.g., 128–256), indicating diminishing returns beyond a certain scale.

Under DP, retrieval accuracy instead peaks at a moderate expected group size and then declines, particularly on RSTPReid (first row of~\Cref{fig:tradeoff}, red curve). The second row shows lower first-iteration SNR for larger groups, both with and without ISA. This observation is consistent with~\Cref{eq:group_count_snr}: at a fixed batch size, larger groups mean fewer clipping units. These measurements support the trade-off in~\Cref{the:tradeoff}.
In these experiments, expected group sizes of approximately 4--16 give the best observed balance between local negative comparisons and signal accumulation under DP. ISA also raises the initial SNR, especially for small groups.

\subsection{Generalization across Training Strategies}
\label{subsec:generalization}

\subsubsection{Fine-tuned Layers} 
To assess the compatibility of \toolname{} with the common parameter-efficient fine-tuning (PEFT) strategy LoRA, we conduct experiments where only specific attention projections (Query \(Q\), Key \(K\), or Value \(V\)) are updated during fine-tuning, while the rest of the model remains frozen.

\begin{table}[!t]
\renewcommand{\arraystretch}{1.1}
\setlength{\tabcolsep}{5.5pt}
    \centering
    \caption{Retrieval accuracy {(\%)} when fine-tuning different layers of the models using LoRA. Q+V is used in our main experiments.}
    \label{tab:lora_qkv}
    \resizebox{0.48\textwidth}{!}{
    \begin{tabular}{l*{8}{c}}
    \toprule
    \multirow{2}{*}{Layer} & \multicolumn{2}{c}{{PEDES}} & \multicolumn{2}{c}{{RSTPReid}} & \multicolumn{2}{c}{{Fashion}} & \multicolumn{2}{c}{{ROCO}} \\
    \cmidrule(lr){2-9}
     & I$\rightarrow$T & T$\rightarrow$I & I$\rightarrow$T & T$\rightarrow$I & I$\rightarrow$T & T$\rightarrow$I & I$\rightarrow$T & T$\rightarrow$I \\
    \midrule
    Q & 56.8 & 40.6 & 30.6 & 27.8 & 17.9 & 16.1 & 32.4 & 34.0\\
    K & 55.8 & 40.2 & 30.6 & 26.8 & 18.5 & 16.7 & 33.0 & 33.1\\
    V & 64.5 & 47.8 & 36.8 & 33.5 & 25.9 & 23.7 & \textbf{37.9} & 37.7\\
    Q+K & 58.9 & 41.9 & 31.9 & 28.8 & 20.2 & 18.3 & 33.1 & 33.6\\
    Q+V & 65.9 & 48.1 & 38.0 & 34.5 & 26.3 & 24.7 & 37.5 & \textbf{38.4}\\
    K+V & 66.1 & 47.7 & \textbf{40.2} & 34.0 & \textbf{26.7} & \textbf{24.9} & 37.5 & 38.2\\
    Q+K+V & \textbf{66.7} & \textbf{48.4} & 38.9 & \textbf{35.2} & 26.5 & 24.7 & 37.2 & 37.3\\
    \bottomrule
\end{tabular}
}
\end{table}

\Cref{tab:lora_qkv} shows that \toolname{} remains effective when only a small subset of attention parameters is fine-tuned. Updating only the \(V\) projections reaches 64.5\% I→T on {PEDES}, compared with 56.8\% for \(Q\) and 55.8\% for \(K\). The combinations that include \(V\) (\(Q+V\), \(K+V\), and \(Q+K+V\)) further improve accuracy and perform comparably. 

\subsubsection{Model Architectures} 
\label{subsubsec:modelarch}
In~\Cref{subsec:overall_results}, we implement \toolname{} using ViT-B32, which contains approximately 151M parameters, for multi-modal contrastive learning. To evaluate the generalization capability of \toolname{} across mainstream vision encoder architectures, we also implement it with three other widely used ViT variants. ViT-B16 has the same size as ViT-B32, with 12 Transformer layers and a hidden dimension of 768, but uses smaller $16\times16$ patches, which provide finer spatial resolution at a higher computational cost. ViT-L14 has 24 Transformer layers, a hidden dimension of 1024, and approximately 428M parameters, and uses $14\times14$ patches on $224\times224$ inputs. ViT-L14-H has the same architecture as ViT-L14 but is pre-trained at a higher input resolution of $336\times336$, which captures finer visual details. All models are trained under \(\epsilon=10\).

\Cref{fig:model_arch} shows that \toolname{} works with all four ViT variants and that its accuracy increases from ViT-B32 to ViT-B16 to ViT-L14 on every dataset. For example, the I→T accuracy on {PEDES} rises from 65.4\% with ViT-B32 to 81.2\% with ViT-L14-H.

\subsection{Hyper-parameter Analysis}
\label{subsec:hyperparameters}

This section explores how the performance of \toolname{} is affected by the privacy budget $\epsilon$, the number of augmentations $N_a$, and the augmentation methods.

\begin{figure}[!t]
    \centering
    \setlength{\abovecaptionskip}{0pt}
    \includegraphics[width=0.98\linewidth]{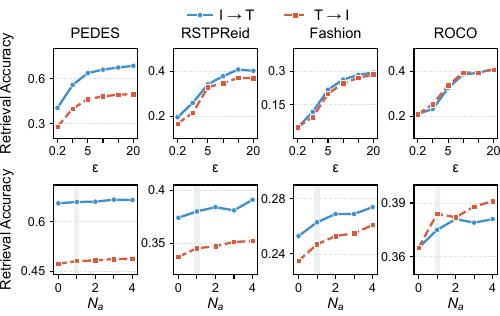}
    \caption{Retrieval accuracy of \toolname{} under different {privacy budgets $\epsilon$ (first row) and numbers of augmented views per negative record $N_a$ (second row)}. $N_a=1$ (in gray) is adopted in our main experiments.}
    \label{fig:epsilon}
\end{figure}

\subsubsection{Privacy Budget} \Cref{fig:epsilon} (first row) illustrates the retrieval accuracy of \toolname{} under varying privacy budgets \((\epsilon = \{0.2, 1.0, 5.0, 10.0, 15.0, 20.0\})\) across four datasets. The results present a consistent trend: performance improves as \(\epsilon\) increases, i.e., as privacy constraints are relaxed, but with diminishing returns beyond \(\epsilon \approx 5.0\).

On all datasets, both I→T and T→I retrieval accuracy rise sharply from \(\epsilon = 0.2\) to \(\epsilon = 5.0\), indicating that extremely tight privacy (e.g., \(\epsilon < 1.0\)) severely limits model utility. Beyond \(\epsilon = 5.0\), gains become marginal. Moderate budgets ($\epsilon\geq5$) thus already give close to the highest accuracy in our experiments.

\subsubsection{Number of Augmentations}
\Cref{fig:epsilon} (second row) evaluates the effect of varying the number of augmented views per negative record \(N_a\) (from 0 to 4) on retrieval accuracy across four datasets. Augmentation consistently improves performance. Even adding one augmented view per negative record (\(N_a=1\)) boosts both I→T and T→I accuracy compared to no augmentation (\(N_a=0\)). The gains are most pronounced on {Fashion} and {ROCO}, where T→I accuracy increases by 2.6 points at \(N_a=4\). Across all datasets, T→I retrieval benefits more from augmentation than I→T, with the exception of {RSTPReid}.
This is likely because text representations are more sparse, making them more sensitive to the quantity of samples.

\begin{table}[!t]
\renewcommand{\arraystretch}{1.1}
\setlength{\tabcolsep}{6pt}
\centering
\caption{Retrieval accuracy (\%) of image-to-text (I$\rightarrow$T) and text-to-image (T$\rightarrow$I) retrieval with different augmentations on PEDES. Sentence Swap and Random Crop are used in our main experiments.}
\label{tab:aug_type}
\resizebox{0.48\textwidth}{!}{
    \begin{tabular}{l*{6}{c}}
    \toprule
    \multirow{2}{*}{\textbf{Augmentation}} & \multicolumn{2}{c}{{Random Crop}} & \multicolumn{2}{c}{{Color Jitter}} & \multicolumn{2}{c}{{Gaussian Blur}} \\
    \cmidrule(lr){2-7}
     & I$\rightarrow$T & T$\rightarrow$I & I$\rightarrow$T & T$\rightarrow$I & I$\rightarrow$T & T$\rightarrow$I \\
    \midrule
    Sentence Swap      & 65.9 & 48.1 & 65.6 & 47.9 & 65.6 & 47.7 \\
    Word Swap    & 65.8 & 48.4 & 65.2 & 48.0 & 65.2 & 48.0 \\
    Word Delete   & 65.6 & 47.8 & 65.6 & 47.8 & 65.1 & 47.5 \\
    \bottomrule
    \end{tabular}
}
\end{table}

\subsubsection{Different Augmentations}
\label{subsubsec:diffaug}
In multi-modal training, ISA generates negative views with augmentations that perturb the input while preserving its semantics. We consider three image augmentations that are widely used in contrastive learning~\cite{simclr}. Random Crop crops a random region, resizes it to $224\times224$, and applies random horizontal flipping. Color Jitter randomly perturbs brightness, contrast, saturation, and hue. Gaussian Blur applies a Gaussian filter with a randomly sampled kernel size and standard deviation. We also consider three text augmentations based on the swap and deletion operations commonly used in text data augmentation~\cite{wei2019eda}. Sentence Swap swaps two sentences with probability 0.5 when the text contains at least two sentences. Word Swap swaps adjacent words with probability 0.3. Word Delete deletes non-essential words with probability 0.2.

We evaluate the sensitivity of \toolname{} to all nine image--text combinations on PEDES under $\epsilon=10$. As shown in Table \ref{tab:aug_type}, the retrieval performance remains relatively stable across all combinations, with variations within 1.0\%. \toolname{} is thus insensitive to the choice of augmentation. Random Crop with Sentence Swap gives the highest I$\rightarrow$T accuracy (65.9\%), so we use it by default.

\section{Discussion and Limitations}
\label{subsec:limitations}

\noindent \textbf{Choosing the number of groups.} 
Because the privacy cost of \toolname{} does not depend on $K$ (\Cref{the:end_to_end_privacy}), $K$ only affects utility. Without noise, larger groups are better, whereas under DP the best expected group size lies between about 4 and 16 in all our experiments (\Cref{fig:toy_k,fig:tradeoff}). \Cref{the:tradeoff} and the toy experiment further indicate that more noise calls for more groups. We therefore use $S=16$ by default and suggest smaller groups for smaller privacy budgets. Selecting $K$ by repeated training on the private data would consume additional privacy budget, so $K$ should be fixed in advance or chosen on public data.

\noindent \textbf{Data-dependent grouping.} 
Grouping similar records together could provide harder negatives. However, such a grouping depends on the records, so adding one record could move other records between groups and break the sensitivity bound in~\Cref{the:group_grad}. Our analysis therefore requires group marks drawn independently of the data, and private data-dependent grouping is left for future work.

\noindent \textbf{Scope of the utility analysis.} 
\Cref{the:tradeoff} analyzes the update in~\Cref{eq:group_count_noisy_gradient} without ISA and assumes that every nonzero group gradient is clipped. Our main experiments use Adam or LARS and ISA, so the theorem explains the trade-off in $K$ but does not predict the accuracy of these runs. The privacy guarantee in~\Cref{the:end_to_end_privacy} does not rely on these assumptions.

\noindent \textbf{Gap to non-private training.} 
\toolname{} narrows but does not close the gap to non-private training. At $\epsilon=10$, for example, its I$\rightarrow$T accuracy on PEDES is 65.9\% compared with 87.9\% without DP (Appendix~\ref{subsec:non-private}).

\noindent \textbf{Large-scale pre-training.} 
Although our method scales well with the batch size, limited compute prevented us from running large-scale contrastive pre-training, which we leave for future work.

\noindent \textbf{More contrastive learning methods.} 
Our current implementation focuses on the InfoNCE loss. Recent advances have introduced alternative frameworks such as BYOL~\cite{byol} and SigLIP~\cite{siglip}, which reduce reliance on negative samples or adopt different optimization strategies. Extending \toolname{} to these objectives is another direction for future work.

\section{Related Work}
\label{sec:related_work}

\subsection{Private Contrastive Learning} 
\label{subsec:private_cl}

DP has been extensively studied in the context of supervised learning with decomposable losses (the loss of each sample is independent of others), where DP-SGD~\cite{dpsgd} and its refinements (such as improved privacy accounting~\cite{rdp,prv} and more effective clipping strategy~\cite{autoclip}) have become standard tools~\cite{wei2023dpmlbench}. However, contrastive learning, which is widely used for self-supervised representation learning, poses unique challenges for private training. Due to its non-decomposable loss structure, where the loss for each sample depends on other samples in the same batch (e.g., via negative pairs), its gradient sensitivity is difficult to bound. Early attempts to apply DP to contrastive-like objectives focused on pairwise losses in convex settings~\cite{dppairlearning1, dppairlearning2}, but these methods are evaluated only on small tabular datasets and do not scale to deep contrastive models or multi-modal tasks. Later works extended the analysis to non-convex settings~\cite{dpcl1} by assuming knowledge of a global Lipschitz constant of the loss, which heavily depends on the encoder and is generally intractable to compute~\cite{hard_lips}.

To circumvent these difficulties, several recent studies avoid using the original non-decomposable contrastive losses entirely. Xu et al.~\cite{dpcl2} adopt user-level DP but reformulate representation learning as a supervised classification task, requiring labeled data and forfeiting the benefits of self-supervision. Yu et al.~\cite{vip} train a vision foundation model under DP by replacing InfoNCE with an instance-separable loss, the Mask AutoEncoder~\cite{MAE} loss, that decouples sample interactions. Similarly, Sander et al.~\cite{dpcaption} obtain DP representations through training a DP caption model. Li et al.~\cite{dpcl3} propose perturbing the similarity matrix directly and computing a noisy contrastive loss, yet their privacy analysis relies on the strong assumption that each individual gradient has bounded sensitivity, which is hard to guarantee for deep encoders. Although recent efforts have sought to enable private contrastive learning while preserving the original loss structure~\cite{batch-level,pair-level}, they still suffer from severe utility degradation.

\subsection{Group Clipping} 
\label{subsec:group_clipping}

Microbatch clipping aggregates gradients within subsets of a batch and clips each subset contribution before the final aggregation~\cite{mcmahan2018general,ponomareva2023dpfy}. Its record-level sensitivity can be bounded by $2C$ when a record affects only one clipped subset contribution~\cite{ponomareva2023dpfy}. Liu et al.~\cite{group_clliping_cls} and van Dijk et al.~\cite{van2022generalizing} also study clipping at the level of sample groups or batches. He et al.~\cite{group_clliping_layer} and Bu et al.~\cite{bu2023accuracy} instead group model parameters and clip their associated gradients. We build on clipping at the level of sample groups. For global InfoNCE, grouping samples for clipping still allows one record to affect every group through shared negatives. \toolname{} jointly restricts contrastive interactions and clips gradients using the same partition, confining each record's influence to one contribution while accumulating signals across groups. Our contribution is this joint objective--clipping design for private contrastive learning.

\subsection{Privacy Attacks on Contrastive Encoders}
Embedding models trained with contrastive or other self-supervised objectives leak information about their training data. Song and Raghunathan~\cite{attackcl1} show that text embeddings enable membership and attribute inference and can be partially inverted, and Pan et al.~\cite{pan2020privacy} recover sensitive keywords from the sentence embeddings of general-purpose language models. EncoderMI~\cite{attackcl2} infers whether an image was in the pre-training set of an encoder from the similarity of the embeddings of its augmented views, and later attacks extend membership inference to visual self-supervised encoders~\cite{attackcl4}, person re-identification models~\cite{attackcl3}, and vision-language models~\cite{hu2025vlmmia}. Stronger attacks based on likelihood ratios~\cite{carlini2022lira} or on informed adversaries that reconstruct training records~\cite{balle2022reconstructing} keep raising the bar for empirical defenses. DP bounds the success of membership inference and reconstruction for every record regardless of the attack, which motivates DP training of contrastive encoders.

\section{Conclusion}
\label{sec:conclusion}

We have presented \toolname{}, a DP contrastive learning framework that jointly localizes contrastive comparisons and clips group gradients to control each record's influence while enabling signal accumulation. By restricting contrastive interactions to disjoint groups and independently clipping each group-local gradient before aggregation, \toolname{} confines each record's influence to a single clipping unit. This joint design maintains a sensitivity bound of $2C$ independent of the number of groups while allowing multiple clipped group contributions to accumulate. Intra-group sample augmentation further enriches local negative views without increasing gradient sensitivity or incurring additional privacy cost. Experiments across eight datasets demonstrate substantial utility improvements over existing DP contrastive learning methods in both uni-modal and multi-modal settings, supported by ablations of the two complementary components of BGC.

\clearpage

\bibliographystyle{ieeetr}
\bibliography{bib}

\appendices

\setcounter{section}{0}
\renewcommand\thesection{\Alph{section}}

\begin{table*}[!t]
    \renewcommand{\arraystretch}{1.1}
    \centering
    \caption{Hyper-parameters under $\epsilon=10$. The training and privacy settings are shared by all methods, so that they are compared under the same privacy budget. The group settings apply only to \toolname{}.}
    \label{tab:hp_ours}
    \resizebox{1.0\textwidth}{!}{\begin{tabular}{l*{9}{c}}
    \toprule
    \multirow{2}{*}{Hyper-parameter}& \multicolumn{5}{c}{Uni-modal Task} & \multicolumn{4}{c}{Multi-modal Task}\\
    \cmidrule(lr){2-10}
    & {{F-MNIST}} & {{CIFAR-10}} & {{EuroSAT}} & {{Camelyon}} & ImageNet & {{PEDES}} & {{RSTPReid}} & {{Fashion}} & {{ROCO}}\\
    \midrule
    Expected batch size \(\hat{B}\) & 2,048 & 2,048 & 2,048 & 2,048 & 65,536 & 8,192 & 8,192 & 8,192 & 8,192\\
    Training iterations \(T\) & 1,200 & 1,200 & 1,200 & 1,200 & 1,200 & 240 & 160 & 320 & 560 \\
    Learning rate \(\lambda\)& $1\times10^{-3}$ & $1\times10^{-3}$ & $1\times10^{-3}$ & $1\times10^{-3}$ & $6\times10^{-1}$ & $1\times10^{-3}$ & $1\times10^{-3}$ & $1\times10^{-3}$ & $1\times10^{-3}$\\
    Optimizer & Adam & Adam  & Adam  & Adam  & LARS & Adam  & Adam  & Adam  & Adam \\
    Augmentation times \(N_a\) & 1 & 1 & 1 & 1 & 1 & 1 & 1 & 1 & 1\\
    Noise scale \(\sigma\)  & {1.02} & {1.14} & {2.00} & {0.59} & {1.34} & {0.79} & {1.77} & {2.67} & {2.00}\\
    Gradient norm \(C\)  & 1 & 1 & 1 & 1 & 1 & 0.01 &  0.01 &  0.01 &  0.01\\
    Temperature hyper-parameter \(\tau\)  & {0.5} & {0.5} & {0.5} & {0.5} & {0.1} & \(1/100\) & \(1/100\) & \(1/100\) & \(1/100\)\\
    \midrule
    Group number \(K\) & 128 & 128 & 128 & 128 & 256 & 512 & 512 & 512 & 512\\
    Expected group size \(S\) & 16 & 16 & 16 & 16 & 256 & 16 & 16 & 16 & 16\\
    \bottomrule
\end{tabular}}
\end{table*}

\section{Proofs}
\label{apsec:MissingProofs}

We prove the sensitivity bounds in Theorems~\ref{the:group_clipping}, \ref{the:group_grad}, and \ref{the:group_grad_aug}, followed by the end-to-end guarantee in Theorem~\ref{the:end_to_end_privacy}. We consider two neighboring datasets, denoted as $\mathcal{D} = \{(x_i, x_i^+)\}_{i=1}^{N}$ and $\mathcal{D}' = \mathcal{D} \cup \{(x_*, x_*^+)\}$. Thus, $\mathcal{D}'$ differs from $\mathcal{D}$ by exactly one additional positive pair. At each iteration, the outer DP-SGD sampler first selects a batch. Conditional on the sampled batch, each sample pair receives an independent random mark in $\{1,\ldots,K\}$, independently of the sample values. For the sensitivity argument, we couple the sampling and marking randomness across neighboring datasets: common pairs keep their marks, while the additional pair receives a fresh random mark.

Let $\mathcal{G}=\{\mathcal{G}_1,\ldots,\mathcal{G}_K\}$ and $\mathcal{G}'=\{\mathcal{G}'_1,\ldots,\mathcal{G}'_K\}$ denote the resulting groups. Under the shared randomness, all common pairs have identical group assignments. Therefore, adding the additional pair changes the membership of only its assigned group, and all other groups retain the same members. Empty groups are allowed, which is important because the group number is fixed while the sampled batch size may vary.

\noindent\textbf{Proof of Theorem~\ref{the:group_clipping}.} Fix a realization of the sampling and grouping randomness, and write $a_k(\mathcal{D}) = \text{Clip}_C(\nabla_\theta \mathcal{L}_k^{\mathrm{all}}(\mathcal{D}))$, where $\mathcal{L}_k^{\mathrm{all}}$ is the full-batch contrastive loss restricted to the samples assigned to group $k$, and define $a_k(\mathcal{D}')$ analogously. Each clipped term has norm at most $C$, and the full-batch denominator may change when the additional record is present, so we use the worst-case bound
\[
\begin{aligned}
&\bigl\|g_{\text{g\_clip}}(\mathcal D)-g_{\text{g\_clip}}(\mathcal D')\bigr\|_2\\
&\leq \textstyle\sum_{k=1}^{K}\|a_k(\mathcal{D})\|_2
    + \sum_{k=1}^{K}\|a_k(\mathcal{D}')\|_2
\leq 2KC.
\end{aligned}
\]
Taking the maximum over adjacent datasets gives
$\Delta_{g_{\text{g\_clip}}}\leq2KC$.

\noindent\textbf{Proof of Theorem~\ref{the:group_grad}.} We condition on the same realization of the sampling and the grouping marks in the two neighboring executions. Let $k^*$ be the mark assigned to the additional pair. Then
\[
    \mathcal{G}'_k=\mathcal{G}_k \quad (k\neq k^*),
    \qquad
    \mathcal{G}'_{k^*}=\mathcal{G}_{k^*}\cup\{(x_*,x_*^+)\}.
\]
For simplicity, we set $k^*=K$. Since the group-level InfoNCE loss only uses samples within its own group, all unchanged groups have identical contributions in the two executions and cancel from the difference. Let
\[
    h(\mathcal{G}) = \text{Clip}_C\left(\nabla_\theta \mathcal{L}_{\mathrm{NCE}}^{\mathrm{group}}(\mathcal{G})\right)
\]
denote the clipped contribution of a group. Therefore,
\[
    \left\|g_{\text{group}}(\mathcal{D})-g_{\text{group}}(\mathcal{D}')\right\|_2
    = \left\|h(\mathcal{G}_K)-h(\mathcal{G}'_K)\right\|_2.
\]
Following the convention in \Cref{subsubsec:group_clipping}, $h(\mathcal{G})=0$ when $|\mathcal{G}|\leq 1$. We consider two cases. If $|\mathcal{G}_K|\leq 1$, then $h(\mathcal{G}_K)=0$. When $|\mathcal{G}_K|=0$, the updated group is a singleton and its contribution is also zero. When $|\mathcal{G}_K|=1$, the updated group may have a nonzero contribution, but its clipped norm is at most $C$. Hence, in either case,
\[
    \left\|h(\mathcal{G}_K)-h(\mathcal{G}'_K)\right\|_2 \leq C.
\]
If $|\mathcal{G}_K|\geq 2$, both contributions can be nonzero. By the triangle inequality and clipping,
\[
    \left\|h(\mathcal{G}_K)-h(\mathcal{G}'_K)\right\|_2
    \leq \left\|h(\mathcal{G}_K)\right\|_2
        + \left\|h(\mathcal{G}'_K)\right\|_2
    \leq 2C.
\]
Taking the worst case over the two cases yields \(\Delta_{g_{\text{group}}}\leq 2C\). The augmentation terms are generated within the same group and therefore do not alter this argument, and thus the proof for $g_{\text{group}}^{\mathrm{aug}}$ is identical.

\begin{figure*}[!t]
    \centering
    \begin{minipage}[b]{0.48\textwidth}
        \centering
        \includegraphics[width=\linewidth,trim=0 10pt 0 0,clip]
        {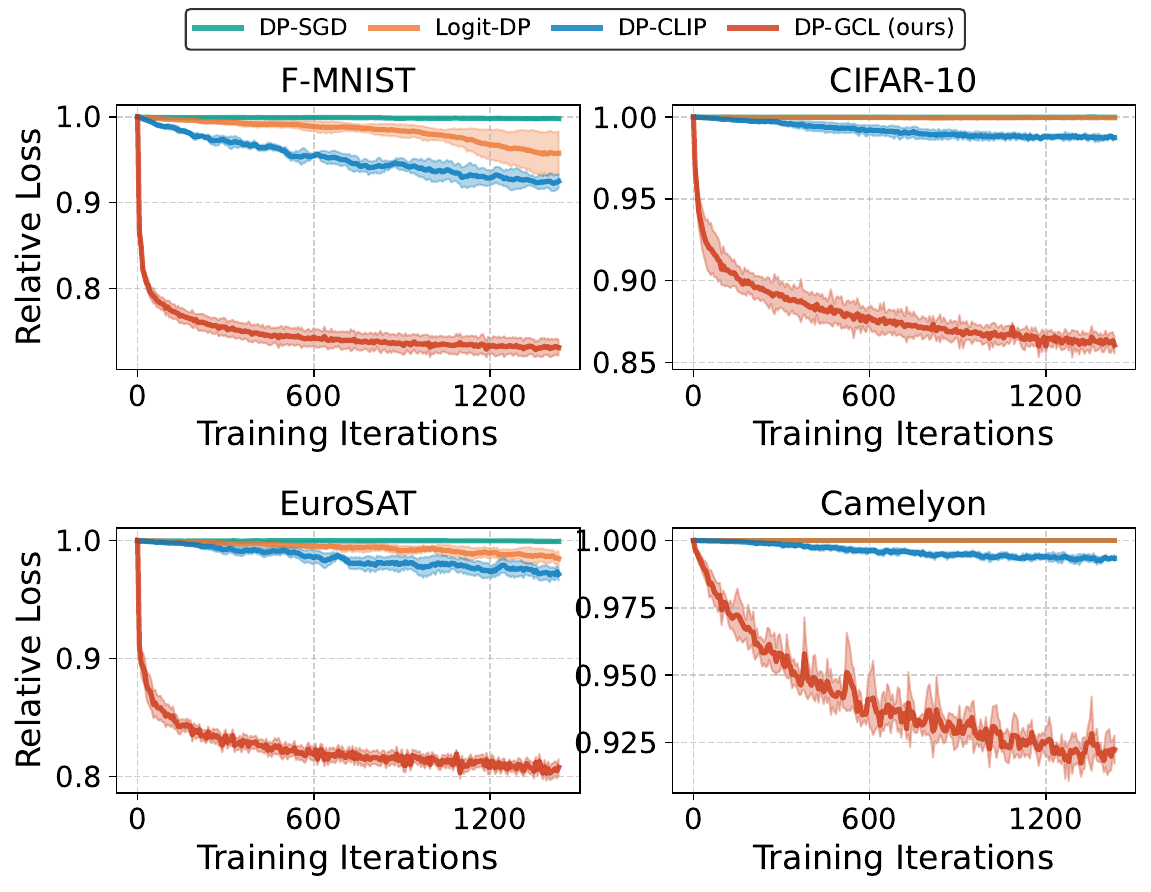}
    \end{minipage}
    \hfill
    \begin{minipage}[b]{0.48\textwidth}
        \centering
        \includegraphics[width=\linewidth,trim=0 10pt 0 0,clip]
        {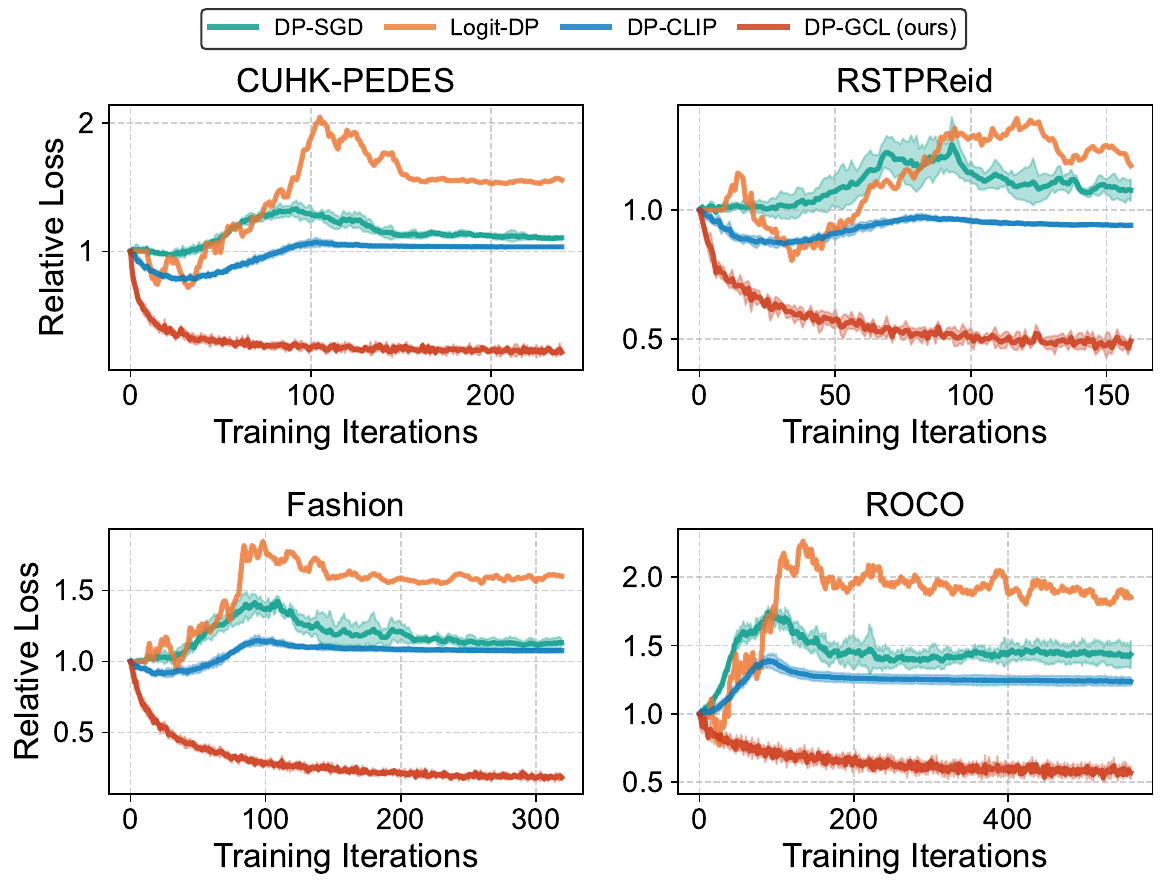}
    \end{minipage}
    \caption{Relative loss (current loss divided by the loss at the first iteration) of three baselines and \toolname{} during training.}
    \label{fig:loss_change}
\end{figure*}

\noindent\textbf{Proof of
Theorem~\ref{the:end_to_end_privacy}.}
Fix the initialization and the preceding noisy
updates, so that the current model parameters and
optimizer state are fixed. The initialization can
be included in the transcript at no privacy cost,
since it does not access the private training set.
All distributions below are conditional on this
same history.

Consider neighboring datasets
$\mathcal D'=\mathcal D\cup\{x^\ast\}$.
Let $\omega$ collect the grouping marks and
augmentation seeds for the records in $\mathcal D'$
at the current iteration. To compare the two
executions, use the same marks and seeds for their
common records, and an additional mark and seeds
for $x^\ast$. This coupling preserves the marginal
sampling laws of the marks and seeds in each
execution. For fixed $\omega$, let
$f_\omega(\mathcal B)$ denote the unnormalized
sum of clipped group gradients.
By the group-local computation and
Theorems~\ref{the:group_grad}
and~\ref{the:group_grad_aug},
\begin{equation}
\label{eq:app_conditional_sensitivity}
\|f_\omega(\mathcal B\cup\{x^\ast\})
-f_\omega(\mathcal B)\|_2\le 2C.
\end{equation}

Let $\mathcal B$ be a Poisson sample from
$\mathcal D$, drawn independently of $\omega$, and write
$P_{\omega,\mathcal B}
=\mathcal N(f_\omega(\mathcal B),
4C^2\sigma^2\mathbb{I}_d)$ and
$P_{\omega,\mathcal B}^{+}
=\mathcal N(f_\omega(\mathcal B\cup\{x^\ast\}),
4C^2\sigma^2\mathbb{I}_d)$.
Conditional on $\omega$, the distributions of
the unnormalized noisy gradient sum on $\mathcal D$
and $\mathcal D'$ satisfy
\begin{equation}
\begin{aligned}
P_\omega
&=\mathbb E_{\mathcal B}
[P_{\omega,\mathcal B}],\\
Q_\omega
&=\mathbb E_{\mathcal B}
\left[(1-q)P_{\omega,\mathcal B}
+qP_{\omega,\mathcal B}^{+}\right].
\end{aligned}
\end{equation}

By~\Cref{eq:app_conditional_sensitivity},
$f_\omega/(2C)$ has sensitivity at most one.
Dividing the noisy sum by $2C$ also gives Gaussian
noise with standard deviation $\sigma$.
Applying the sampled-Gaussian reduction
in~\cite[Theorem~4]{sgm} to the common-subsample
mixtures above gives, in both directions,
\begin{equation}
D_\alpha(P_\omega\|Q_\omega),
D_\alpha(Q_\omega\|P_\omega)
\le
\varepsilon_{\mathrm{SGM}}(\alpha;q,\sigma).
\end{equation}
This reduction applies to general
bounded-sensitivity functions and does
not require per-record gradient decomposition.

We next average over the grouping marks and
augmentation seeds. Their draws are fresh at the
current iteration and independent of the record
values, so the coupling uses the same distribution
of $\omega$ in both executions, even conditional on
the preceding history. Define
$P=\mathbb E_\omega[P_\omega]$ and
$Q=\mathbb E_\omega[Q_\omega]$.
Joint convexity of the exponentiated R\'enyi
divergence then implies
\begin{equation}
\begin{aligned}
e^{(\alpha-1)D_\alpha(P\|Q)}
&\le
\mathbb E_\omega
\bigl[
e^{(\alpha-1)D_\alpha(P_\omega\|Q_\omega)}
\bigr]\\
&\le
e^{(\alpha-1)
\varepsilon_{\mathrm{SGM}}(\alpha;q,\sigma)}.
\end{aligned}
\end{equation}
The reverse direction follows identically.

Thus, each iteration satisfies
$(\alpha,\varepsilon_{\mathrm{SGM}}
(\alpha;q,\sigma))$-RDP, conditional on
the preceding transcript.
Adaptive RDP composition over $T$ steps
establishes~\Cref{eq:end_to_end_rdp}.
Normalization by the fixed $K$ and
subsequent optimizer updates are
post-processing. Finally, the RDP-to-DP
conversion of~\cite{privacyAccount2} gives
\Cref{eq:end_to_end_dp},
completing the proof.

\subsection{Analysis of the Number of Groups}
\label{ap:group_count_analysis}

\noindent\textbf{Proof of \Cref{eq:group_error}.} Let $w_{i,j}=e^{s_{i,j}/\tau}\in[e^{-1/\tau},e^{1/\tau}]$, and let $\ell_i(\mathcal A)$ be the $i$-th term of InfoNCE with its denominator taken over a set of records $\mathcal A\ni i$. Its gradient is $(\phi_i(\mathcal A)-\nabla_\theta s_{i,i})/\tau$ with $\phi_i(\mathcal A)=\sum_{j\in\mathcal A}w_{i,j}\nabla_\theta s_{i,j}/\sum_{j\in\mathcal A}w_{i,j}$, and its norm is at most $2G/\tau$. By the definitions of $\nabla F$ and $\bar g_k$, a group of size $M\ge1$ satisfies
\[
\bar g_k-\nabla F=\frac1M\sum_{i\in\mathcal G_k}\Bigl[\nabla_\theta\ell_i(\mathcal D)-\nabla F+\frac{\phi_i(\mathcal G_k)-\phi_i(\mathcal D)}{\tau}\Bigr].
\]
Given $M$, the group is a uniformly random subset of $\mathcal D$. The first two terms sum to the error of a sample mean, with an expected norm of at most $2G/(\tau\sqrt M)$. For the last term, the numerator and the denominator of $\phi_i(\mathcal G_k)$, divided by $M$, are sample means of terms bounded by $e^{1/\tau}G$ and $e^{1/\tau}$, in which record $i$ has weight $1/M$. They are within $\sqrt{5/M}$ times these bounds of the means over $\mathcal D$ in root mean square, {since the variance is at most $1/M$ and the squared bias at most $4/M^2$ times the squared bounds,} and the denominator is at least $e^{-1/\tau}$, so $\mathbb E[\|\phi_i(\mathcal G_k)-\phi_i(\mathcal D)\|_2\mid M]\le2\sqrt5e^{2/\tau}G/\sqrt M$. Hence $\mathbb E[\|\bar g_k-\nabla F\|_2\mid M]\le c/\sqrt M$ with $c=(2\sqrt5e^{2/\tau}+2)G/\tau$. Finally, $M\sim\mathrm{Binomial}(N,q/K)$ has mean $S$, and $\bar g_k=0$ if $M=0$. Since $1/M\le2/(M+1)$, $\mathbb E[1/(M+1)]\le1/S$, and $\Pr[M=0]\le e^{-S}\le1/\sqrt{2eS}$, Jensen's inequality gives~\Cref{eq:group_error} with $\gamma=\sqrt2c+\sqrt{2/e}\,G/\tau$.

\noindent\textbf{Proof of \Cref{the:tradeoff}.}
Let $n_k=\nabla_\theta\mathcal L_k/\|\nabla_\theta\mathcal L_k\|_2$ for every group with a nonzero gradient and $n_k=0$ otherwise, and set $\bar g_k=0$ for empty groups, so that $n_k$ is the direction of $\bar g_k$. Let $m_t=K^{-1}\sum_k n_k$ at iteration $t$. Since every nonzero group gradient is clipped, $u_k=Cn_k$ and $\widetilde g_K=Cm_t+\xi/K$. The $L$-smoothness of $F$ gives $F(\theta_{t+1})\le F(\theta_t)-\lambda\langle\nabla F(\theta_t),\widetilde g_K\rangle+\frac{L\lambda^2}{2}\|\widetilde g_K\|_2^2$.
The noise $\xi$ has zero mean and is independent of $m_t$, and $\|m_t\|_2\le1$. Hence $\mathbb E_\xi[\widetilde g_K]=Cm_t$ and $\mathbb E_\xi\|\widetilde g_K\|_2^2=C^2\|m_t\|_2^2+4C^2\sigma^2d/K^2\le C^2(1+4\sigma^2d/K^2)$.
For any vectors $g$ and $v\neq0$, we have $\langle g,v/\|v\|_2\rangle=\|v\|_2+\langle g-v,v/\|v\|_2\rangle\ge\|g\|_2-2\|g-v\|_2$. For $v=0$, it holds with $v/\|v\|_2$ replaced by $0$. Applying it with $g=\nabla F(\theta_t)$ and $v=\bar g_k$, and taking the expectation over the batch and the group assignment, gives by~\Cref{eq:group_error} $\mathbb E\langle\nabla F(\theta_t),m_t\rangle\ge\mathbb E\|\nabla F(\theta_t)\|_2-2\gamma\sqrt{K/\hat B}$.
Taking expectations, combining these bounds, and averaging over $t$ yields
\[
\frac1T\sum_{t=0}^{T-1}\mathbb E\|\nabla F(\theta_t)\|_2
\le\frac{F_0}{\lambda CT}
+\frac{L\lambda C}{2}\Bigl(1+\frac{4\sigma^2d}{K^2}\Bigr)
+2\gamma\sqrt{\frac{K}{\hat B}}.
\]
Minimizing the right side over $\lambda$ gives~\Cref{eq:tradeoff_bound}. When $K\le2\sigma\sqrt d$, the first term of~\Cref{eq:tradeoff_bound} lies between $a/K$ and $\sqrt2a/K$ with $a=2\sigma\sqrt{2LF_0d/T}$, and $a/K+2\gamma\sqrt{K/\hat B}$ is minimized at $K^\star=(a\sqrt{\hat B}/\gamma)^{2/3}=\Theta((\sigma^2d\hat B/T)^{1/3})$.

\medskip
\noindent\textbf{Gradient Signal-to-Noise Ratio.} 
\Cref{eq:group_count_snr} follows from $\|\sum_ku_k\|_2=KC\rho_K$ and $\mathbb E_\xi\|\xi\|_2^2=4C^2\sigma^2d$.
This identity does not require independent
group gradients or saturated clipping. Empty and singleton groups contribute zero gradients, including under ISA, and therefore do not increase the accumulated signal.

\section{Implementation Details}
\label{apsec:id}

\label{ap:subsec:ours}

\Cref{tab:hp_ours} lists the hyper-parameters, which are shared by \toolname{} and all baselines. All experiments use $N_a=1$, and all datasets except ImageNet use Adam with a learning rate of $10^{-3}$. The clipping norm is $C=1$ for uni-modal and $C=0.01$ for multi-modal tasks, and the number of iterations $T$ is set per dataset to balance the privacy budget and convergence. For the larger ImageNet dataset, we use an expected batch size of 65,536 with $K=256$, the LARS optimizer with a learning rate of $0.6$, and {$\tau=0.1$}. For every dataset, the noise scale $\sigma$ is calibrated with the RDP accountant~\cite{rdp} to the target $\epsilon$ at $\delta=1/(N\log N)$, where $N$ is the size of the training set.

\begin{table}[!t]
\setlength{\tabcolsep}{5.5pt}
\small
\centering
\caption{Membership Inference Attack (MIA) Performance (TPR@1\%FPR, \%) on different datasets.
Lower values indicate stronger privacy protection.}
\setlength{\tabcolsep}{2.4mm}{
\resizebox{0.48\textwidth}{!}{
\begin{tabular}{l*{4}{c}}
\toprule
\textbf{Method} &PEDES &  RSTPReid & Fashion & ROCO \\
\midrule
DP-GCL w/ DP & 2.6±1.7 & 0.0±0.0 & 1.4±0.9  & 0.0±0.0 \\
DP-GCL w/o DP & 10.4±2.0 & 65.2±8.3 & 40.8±12.8 & 54.2±9.9 \\
\bottomrule
\end{tabular}}}
\label{tab:mia}
\end{table}

\noindent\textbf{Toy Experiment.} For the toy experiment in~\Cref{subsubsec:group_clipping}, we draw 64 cluster centers $\bar z_c\sim\mathcal{N}(0,4\mathbb{I}_{16})$ and two fixed projection matrices $A,A^+\in\mathbb{R}^{32\times16}$ with i.i.d.\ $\mathcal{N}(0,1/16)$ entries. Each pair is generated from a latent vector $z_i=\bar z_{c_i}+0.6\,\eta_i$ with a uniformly drawn cluster $c_i$, as $x_i=Az_i+0.3\,\zeta_i$ and $x_i^+=A^+z_i+0.3\,\zeta_i^+$, where $\eta_i$, $\zeta_i$, and $\zeta_i^+$ are standard Gaussian vectors. We use $N=20{,}000$ training pairs and 1,000 test pairs. The encoders are $\mathcal{E}(x)=Wx$ and $\mathcal{E}^+(x^+)=W^+x^+$ with $W,W^+\in\mathbb{R}^{16\times32}$ initialized with i.i.d.\ $\mathcal{N}(0,1/32)$ entries, so $d=1{,}024$. The similarity follows~\Cref{eq:sim} with temperature $\tau=0.05$. Training uses the group-local objective in~\Cref{eq:simclr_group} without ISA and the update of~\Cref{the:tradeoff}, with Poisson sampling at rate $q=0.0512$ (expected batch size 1,024), $C=0.01$, $T=400$, and $\delta=1/(N\log N)$ in the RDP accountant. The learning rate $\lambda$ is selected among powers of two for each $K$ and $\sigma$ by the mean test accuracy. \Cref{fig:toy_k}(a) and (b) average over 256 draws of a Poisson-sampled batch and a group assignment. In \Cref{fig:toy_k}(a), $\nabla F$ is computed exactly on all $N$ training pairs. \Cref{fig:toy_k}(c) reports the top-1 accuracy of retrieving $x_i^+$ from $x_i$ among the test pairs over five runs with different data and initialization.

\section{More Results}
\label{apsec:more_results}

\subsection{Convergence Efficiency}\label{ap:subsec:convergence}
\Cref{fig:loss_change} shows the relative loss, i.e., the current loss divided by the initial loss, during training under $\epsilon=10$. Because the methods optimize different objectives with different negative sets, the curves describe each method's progress from its own starting point rather than the quality of a common optimum.

On uni-modal tasks, \toolname{} reduces the relative loss to about 0.6--0.7 within the first 50 iterations on F-MNIST and EuroSAT and keeps the lowest relative loss on CIFAR-10 and Camelyon, where the baselines plateau near 0.9--0.95. On multi-modal tasks, \toolname{} also reduces the relative loss faster and more steadily, whereas the baselines show higher variance, slower decay, or even divergence.

\subsection{Membership Inference Attacks}
\label{subsec:mia}
We run a white-box loss-based membership inference attack (MIA)~\cite{dpmia,nasr2019comprehensive} against \toolname{} with and without DP under $\epsilon=10$. As shown in~\Cref{tab:mia}, the TPR at 1\% FPR is at most 2.6\% with DP, whereas it rises to 10.4\%--65.2\% without DP, showing that DP substantially reduces membership leakage even against adversaries with full access to the model parameters.

\begin{table}[!t]
\renewcommand{\arraystretch}{1.1}
\setlength{\tabcolsep}{3.5pt}
    \centering
    \caption{Classification accuracy (\%) of {linear probing} and $k$-NN ($k=3$) under transfer learning settings, where ImageNet is used for pre-training.}
    \label{tab:main_cls_transfer_ap}
    \resizebox{0.48\textwidth}{!}{
    \begin{tabular}{l*{8}{c}}
    \toprule
    & \multicolumn{4}{c}{$\epsilon=1$} & \multicolumn{4}{c}{$\epsilon=10$}\\
    \cmidrule(lr){2-5}\cmidrule(lr){6-9}
    \multirow{2}{*}{\textbf{Method}} & \multicolumn{2}{c}{EuroSAT} & \multicolumn{2}{c}{Camelyon} & \multicolumn{2}{c}{EuroSAT} & \multicolumn{2}{c}{Camelyon} \\
    \cmidrule(lr){2-3}\cmidrule(lr){4-5}\cmidrule(lr){6-7}\cmidrule(lr){8-9}
     & Linear & $k$-NN & Linear & $k$-NN & Linear & $k$-NN & Linear & $k$-NN \\
    \midrule
    Base & 50.5 & 33.4 & 69.7 & 52.4 & 50.5 & 33.4 & 69.7 & 52.4 \\
    DP-SGD~\cite{dpsgd} & 48.5 & 34.6 & 68.7 & 55.6 & 48.9 & 36.2 & 69.7 & 55.6\\
    Logit-DP~\cite{pair-level} & 49.1 & 33.5 & 70.6 & 51.4 & 49.1 & 36.4 & 67.1 & 53.8\\
    DP-CLIP~\cite{batch-level} & 48.1 & 35.1 & 70.6 & 53.4 & 50.0 & 36.4 & 67.1 & 53.8 \\
    \midrule
    \rowcolor{gray!12} \toolname{} (ours) & 50.6 & 37.8 & 69.8 & 56.4 & 53.5 & 44.4 & 69.7 & 59.9 \\
    \bottomrule
\end{tabular}
}
\end{table}

\begin{table}[!t]
\renewcommand{\arraystretch}{1.1}
\setlength{\tabcolsep}{2.5pt}
    \centering
    \caption{The accuracy performance (\%) of \toolname{} under \(\epsilon=10\) and without DP protection.}
    \label{tab:non-private}
\resizebox{0.48\textwidth}{!}{
    \begin{tabular}{l*{8}{c}}
    \toprule
    \multirow{2}{*}{\textbf{Method}} & \multicolumn{2}{c}{{F-MNIST}} & \multicolumn{2}{c}{{CIFAR-10}} & \multicolumn{2}{c}{{EuroSAT}} & \multicolumn{2}{c}{{Camelyon}} \\
    \cmidrule(lr){2-9}
     & Linear & \(k\)-NN & Linear & \(k\)-NN & Linear & \(k\)-NN & Linear & \(k\)-NN \\
    \midrule
    Non-Private & 92.6 & 92.7 & 60.3 & 59.4 & 88.4 & 89.2 & 83.6 & 84.5\\
    \toolname{} & 82.9 & 82.8 & 42.3 & 39.5 & 59.1 & 53.3 & 78.4 & 76.7\\
    \midrule
    \midrule
    \multirow{2}{*}{\textbf{Method}} & \multicolumn{2}{c}{{PEDES}} & \multicolumn{2}{c}{{RSTPReid}} & \multicolumn{2}{c}{{Fashion}} & \multicolumn{2}{c}{{ROCO}} \\
    \cmidrule(lr){2-9}
     & I$\rightarrow$T & T$\rightarrow$I & I$\rightarrow$T & T$\rightarrow$I & I$\rightarrow$T & T$\rightarrow$I & I$\rightarrow$T & T$\rightarrow$I \\
    \midrule
    Non-Private & 87.9 & 70.9 & 53.3 & 47.0 & 51.9 & 51.4 & 58.3 & 59.8\\
    \toolname{} & 65.9 & 48.1 & 38.0 & 34.5 & 26.3 & 24.7 & 37.5 & 38.4\\
    \bottomrule
\end{tabular}
}
\end{table}

\subsection{Performance Gap to Non-Private Baseline}
\label{subsec:non-private}
Table~\ref{tab:non-private} compares the performance of the non-private baseline with that of \toolname{}{} under DP protection (\(\epsilon = 10\)) across eight datasets. The accuracy loss caused by DP varies across datasets. Simpler datasets (e.g., F-MNIST) and larger datasets (e.g., PEDES and Camelyon) generally suffer relatively smaller drops in performance. In contrast, the small and more complex multi-modal dataset Fashion experiences a more significant accuracy loss. Overall, there remains substantial room for improvement in narrowing this performance gap.

\end{document}